%% file: MergerSim.tex
\documentclass[format=acmsmall, review=false]{acmart}
\usepackage{acm-ec-25}
\usepackage{booktabs} 
\usepackage[ruled]{algorithm2e} 

\SetAlFnt{\small}
\SetAlCapFnt{\small}
\SetAlCapNameFnt{\small}
\SetAlCapHSkip{0pt}
\IncMargin{-\parindent}

\usepackage{threeparttable, widetable, subfigure, multirow}

\newtheorem{assumption}{Assumption}
\newcommand{\calF}[0]{\mathcal{F}}
\newcommand{\calH}[0]{\mathcal{H}}
\newcommand{\E}[0]{\mathbb{E}}
\newcommand{\R}[0]{\mathbb{R}}

\setcitestyle{authoryear}

\title[Enhancing the Merger Simulation Toolkit with ML/AI]{Enhancing the Merger Simulation Toolkit with ML/AI}

\author{Harold D. Chiang}
\affiliation{%
  \institution{University of Wisconsin-Madison}
  \city{Madison}
  \state{WI}
  \country{USA}
}
\email{hdchiang@wisc.edu}
\author{Jack Collison}
\affiliation{%
  \institution{University of Wisconsin-Madison}
  \city{Madison}
  \state{WI}
  \country{USA}
}
\email{jcollison@wisc.edu}
\author{Lorenzo Magnolfi}
\affiliation{%
  \institution{University of Wisconsin-Madison}
  \city{Madison}
  \state{WI}
  \country{USA}
}
\email{magnolfi@wisc.edu}
\author{Christopher Sullivan}
\affiliation{%
  \institution{University of Calgary}
  \city{Calgary}
  \state{AB}
  \country{CA}
}
\email{christopher.sullivan1@ucalgary.ca}

\begin{abstract}
This paper develops a flexible approach to predict the price effects of horizontal mergers using ML/AI methods. While standard merger simulation techniques rely on restrictive assumptions about firm conduct, we propose a data-driven framework that relaxes these constraints when rich market data are available. We develop and identify a flexible nonparametric model of supply that nests a broad range of conduct models and cost functions. To overcome the curse of dimensionality, we adapt the Variational Method of Moments (VMM) \citep{bennett2023vmm} to estimate the model, allowing for various forms of strategic interaction. Monte Carlo simulations show that our method significantly outperforms an array of misspecified models and rivals the performance of the true model, both in predictive performance and counterfactual merger simulations. As a way to interpret the economics of the estimated function, we simulate pass-through and reveal that the model learns markup and cost functions that imply approximately correct pass-through behavior. Applied to the American Airlines-US Airways merger, our method produces more accurate post-merger price predictions than traditional approaches. The results demonstrate the potential for machine learning techniques to enhance merger analysis while maintaining economic structure.
\end{abstract}


\begin{document}

\begin{titlepage}

\maketitle

\begin{acks}
   We thank the audiences at the Midwest IO Fest 2024, IIOC 2025, Northwestern, Bristol, and Warwick for helpful discussions. We also thank the anonymous EC referees for their valuable comments and suggestions. Any errors are our own.
\end{acks}

\setcounter{tocdepth}{1} 
\tableofcontents

\end{titlepage}

\section{Introduction}

A central challenge in antitrust policy is assessing the competitive risks of proposed horizontal mergers. As emphasized in the U.S. Horizontal Merger Guidelines, Section 7 of the Clayton Act requires agencies to ``arrest anticompetitive tendencies in their incipiency,'' conducting an assessment of what will likely happen if a merger proceeds as compared to what will likely happen if it does not. While various tools are employed for this analysis, merger simulation based on structural economic models represents a sophisticated quantitative approach that can provide valuable indications of the scale and importance of competition between merging firms. These models typically focus on firms' incentives to change their actions in specific competitive dimensions, such as price, while abstracting from other competitive forces.

The standard merger simulation methodology, developed through seminal contributions  \citep{werden1994effects, nevo2000mergers}, relies on estimating demand elasticities and modeling firm conduct under specific assumptions about competitive behavior. Typically, firms are assumed to engage in Bertrand-Nash price competition with constant marginal costs. While this framework has proven valuable in many applications, retrospective studies indicate that its predictive performance is mixed. A key limitation identified in the literature is that the restrictive assumptions about firm conduct may not adequately capture the complexity of real-world competitive interactions \citep[see e.g.,][]{peters2006evaluating}. While one could alternatively pursue a purely data-driven approach using machine learning methods, such an approach would ignore the fundamental simultaneity in market equilibrium and likely produce poor predictions.\looseness=-1 

This paper proposes a novel approach that maintains the economic structure of merger simulation while relaxing assumptions about firm conduct when sufficiently rich market data are available. Our model maintains the core economic insight that market prices emerge from the simultaneous interaction of demand and supply conditions while relaxing parametric assumptions about both markups and cost functions. The framework respects the structure of oligopoly equilibrium -- where prices and quantities jointly solve a system of demand and supply equations. At the same time, we allow for a flexible, nonparametric specification of how firms set markups as a function of market conditions and the marginal cost they face, nesting both the standard Bertrand-Nash model and alternative models of competition. Moreover, economic theory guides our choice of instrumental variables for handling the endogeneity of prices and market shares. However, the high dimensionality of this problem poses significant estimation challenges that cannot be addressed with classical nonparametric methods.\looseness=-1 

We overcome these challenges by adapting recent advances in machine learning, specifically the Variational Method of Moments (VMM) developed by \cite{bennett2023vmm}. VMM combines deep learning with instrumental variables estimation through a novel min-max formulation of moment conditions. This approach, by incorporating instruments, can thus address the simultaneity inherent in market equilibrium models, where prices and quantities are jointly determined. By parameterizing the supply function as a flexible neural network and optimizing over a growing class of functions, we can use VMM to avoid restrictive parametric assumptions while maintaining the core economic structure of oligopoly competition. 

To motivate the choice of VMM, we contrast it with conventional Nonparametric Instrumental Variable (NPIV) methods. NPIV estimators offer flexibility in capturing complex relationships and possess well-established theoretical properties. However, they suffer from the curse of dimensionality, limiting their practicality in settings with multiple endogenous variables. Like many new machine learning-based methods, VMM mitigates this issue by using parametric models (e.g., neural networks) that adaptively grow in complexity with sample size. This approach preserves NPIV's strengths while mitigating dimensionality challenges. Moreover, in fixed-dimensional parametric settings, VMM coincides with the optimally weighted Generalized Method of Moments (GMM).
A key innovation in our adaptation of VMM addresses a key aspect of our application: while \cite{bennett2023vmm} provides valid element-wise inference for one-dimensional nuisance parameters, our focus is on predicting post-merger prices, a more complex functional of these parameters. To address this, we develop an inference procedure that leverages the numerical delta method, Holm's step-down procedure, and a permutation-based implementation. This ensures computationally feasible standard errors and confidence intervals.

Monte Carlo simulations demonstrate the strong predictive performance of our approach. When the data-generating process follows a standard model of competition (e.g., Bertrand), VMM recovers predictions close to the true structural model. More importantly, when the true model of competition differs from standard assumptions, our method substantially outperforms misspecified structural approaches. The inclusion of demand derivatives as inputs improves prediction accuracy, particularly with larger neural network architectures in data-rich environments. 

Beyond prediction, we address a common criticism of machine learning approaches: their ``black box'' nature. We devise a method to interpret the supply model implied by our VMM estimates by computing numerical approximations of cost pass-through matrices. We find that the pass-through implied by the flexible model closely aligns with that of the true model, demonstrating that our approach not only performs well in prediction but also captures the underlying economic structure of firm conduct.

The airline industry provides an ideal setting to evaluate our methodology. The industry features substantial variation in market structure both across routes and over time, allowing us to learn about competitive conduct under different market conditions. Additionally, several major mergers in recent years provide opportunities for retrospective analysis. We focus on predicting the price effects of the 2013 American Airlines-US Airways merger, comparing our flexible approach to standard merger simulation techniques.

Our empirical analysis yields several key findings. First, the flexible supply-side specification achieves better in-sample fit than the standard Bertrand-Nash model, reducing the passenger-weighted mean squared error in predicted prices by approximately 40\%. More importantly, when predicting post-merger prices in markets affected by the American-US Airways merger, our method produces estimates closer to observed outcomes compared to traditional merger simulation. The passenger-weighted mean squared error in post-merger prices for the Bertrand-Nash model is 365.71, while the passenger-weighted mean squared error for the flexible model is 66.93, a reduction of more than 80\%. The median predicted price increase from our flexible approach is 2.05\%, closely matching the difference-in-differences estimate of 2.92\%, while the standard Bertrand model predicts a median increase of 1.45\%.

Our method's computational requirements are manageable by modern standards, with estimation typically completed within hours on standard hardware. The approach is highly portable across empirical settings, requiring only standard market-level price and quantity data. We provide detailed guidance and code for implementation.

As for any method, there are trade-offs to consider. In particular, the method's applicability depends on the empirical setting. In cases where a proposed merger would create market structures radically different from those observed historically, theory-driven approaches may be more appropriate. However, our simulations show that the method performs well even with moderate sample sizes and when the exact post-merger market structure is not observed in the training data, provided there is sufficient variation in market structure to inform prediction. This makes the approach particularly valuable for industries with rich historical data on market structure changes.



This paper contributes to the IO literature that proposes nonparametric models of market equilibrium. We build our flexible model of supply on the identification results of \cite{bh14}. Similarly to \cite{c20}, who develops a method to estimate demand nonparametrically, our paper proposes a method for nonparametric estimation of the supply side. With similar motivation, recent papers by \cite{gh20b} (proposing a linear approximation of the markup function) and \cite{op24} (using the revelation principle) also develop methods to bring more flexible supply models to data. We complement these approaches by proposing a method that leverages advances in ML/AI coupled with a nonparametric structure.

A growing literature seeks to combine ML/AI methods with structure coming from economic theory. \cite{farrell2020deep} provide a theoretical foundation for using deep learning in two-stage econometric procedures, allowing for rich heterogeneity while maintaining model interpretability. \cite{kaji2023adversarial} introduce ``adversarial estimation,'' which formulates the estimator as a minimax problem between a generator and discriminator, achieving parametric efficiency under correct specification. \cite{wei2024estimating}'s Neural Net Estimator (NNE) approach trains neural networks to directly estimate structural parameters from data moments. Similar to these methods, we aim to leverage the flexibility and computational efficiency of machine learning while preserving the economic structure and interpretability of traditional models. We differ in the specific implementation and application: we formulate a nonparametric structural model of market equilibrium, and estimate it using the VMM formulation of \cite{bennett2023vmm}.

There has been increasing interest in employing general function approximation techniques - such as deep neural networks and random forests - in a unified approach to IV problems. Several studies, including those by \cite{dikkala2020minimax,lewis2018adversarial,liao2020provably},  and \cite{zhang2023instrumental} have advanced this literature. Our study is based on VMM for two main reasons. First, \cite{bennett2023vmm} derived its asymptotic distribution, which enables us to construct confidence intervals for the unknown supply-side function of interest. Second, because VMM subsumes optimally-weighted two-step GMM as a special case, it represents a natural evolution from the existing parametric estimation methods in IO.

This paper also contributes to the literature that seeks to enhance the standard merger simulation toolkit with more flexible and data-driven procedures. \cite{benkard2010simulating} and \cite{BrueggePolicyFcns24} develop methods based on policy functions. In comparison, our method considers a more restrictive static structure and does not predict post-merger entry or exit. The ML/AI tools we adopt, however, allow us to be nonparametric, and the explicit equilibrium model we develop allows us to quantify welfare and firm profit post-merger.

The remainder of the paper is organized as follows. Section \ref{sec:background} presents the standard merger simulation framework that serves as our benchmark. Section \ref{sec:flexible-model} develops our flexible approach and discusses identification. Section \ref{sec:vmm} describes the VMM estimation procedure and derives its statistical properties. Section \ref{sec:simulation} presents Monte Carlo evidence comparing our method to standard approaches. Section \ref{sec:application} applies our methodology to airline mergers and discusses the results. Section \ref{sec:conclusion} concludes with a discussion of limitations and directions for future research.

\section{Background: The Standard Merger Simulation Toolkit} \label{sec:background}

The standard merger simulation toolkit, developed building on the seminal contributions of \cite{bp93,hausman1994competitive,werden1994effects}, is based on a demand system and an assumed model of conduct. In this section, we introduce the market environment we will consider in the paper and summarize the basic elements of the merger simulation toolkit.\footnote{For a more detailed description of the toolkit, see \citep{ms21}.}\\

\noindent\textbf{Data-generating Process:} We observe pre-merger data on a set of products $\mathcal{J}$ offered by firms across a set of $\mathcal{T}$ markets. The researcher observes prices $p_{jt}$, market shares $s_{jt}$, a vector of product characteristics ${x}_{jt}$ that enter demand, a vector of cost shifters $\text{w}_{jt}$, and ownership structure $\mathcal{H}_t$ for each product-market pair $(j,t)$. For any variable ${y}_{jt}$, we denote ${y}_t$ as the vector of values in market $t$. We assume, across all markets, that the true demand system is given by:
\begin{equation}
    {s}_t = \mathscr{s}({p}_t, {x}_t, {\xi}_t, \theta^D_0).
\end{equation}

The term ${\xi}_t$ is a vector of unobservable product characteristics and $\theta^D_0$ is the vector of true demand parameters. Data on endogenous prices and quantities are generated by equilibrium in market $t$. On the supply side, firm behavior is characterized by a set of first-order conditions for the firms' profit maximization problems:
\begin{equation}
    {p}_t = \Delta_{0}({p}_t, {s}_t, \mathcal{H}_t, \theta^D_0) + {c}_{0t}.
\end{equation}

We denote $\Delta_{0}({p}_t, {s}_t, \mathcal{H}_t, \theta^D_0)$ as the vector of true markups and ${c}_{0t}$ as the vector of true marginal costs in market $t$. Therefore, an assumption on conduct, which determines the form of the markup function, and an assumption on how costs are determined are the key ingredients of the supply-side model.

A standard assumption on conduct is  Bertrand-Nash, whereby firms are playing a complete information pricing game, and the markup function is given by $\Delta_0({p}_t, {s}_t, \mathcal{H}_t, \theta^D_0) = \left( \mathcal{H}_t \odot \frac{\partial {s}_t}{\partial {p}_t} \right)^{-1} {s}_t$. Other assumptions are possible, including Cournot-Nash, Nash Bargaining, or profit weight models.  

Firms' costs are generated by some cost function $\mathscr{c}$, or ${c}_{0t}=\mathscr{c}(q_{jt}, \text{w}_{jt}, \omega_{0jt})$, where $\text{w}_{jt}$ and $\omega_{0jt}$ are, respectively, observed and unobserved cost-shifter variables. The canonical case in the literature is costs that are constant in quantities.\\

\noindent\textbf{The Merger Simulation Problem:} Suppose the researcher wants to predict prices following a horizontal merger that results in a post-merger ownership matrix $\tilde{\mathcal{H}}_t$ in market $t$. In the standard merger simulation toolkit, the researcher uses a combination of data and assumptions on the data-generating process to assess the effects of the merger. 

This analysis typically proceeds in three steps. First, the researcher formulates a parametric demand model and estimates demand primitives $\hat \theta^D$.\footnote{We abstract here from the possibility of demand misspecification and assume the researcher correctly specifies demand.} Second, under an assumption on conduct, whereby the researcher specifies a model $m$ with corresponding markup function $\Delta_m$, and the researcher inverts marginal costs ${c}_{mt} = {p}_t - \Delta_{m}({p}_t, {s}_t, \mathcal{H}_t, \hat \theta^D)$. Third, the researcher computes post-merger prices $\tilde{{p}}_t$ under the post-merger ownership structure $\tilde{\mathcal{H}}_t$. Post-merger prices solve a fixed point:
\begin{equation}\label{eq:postm_pred}
    \tilde{{p}}_t = {c}_{mt} + \Delta_{m}(\tilde{{p}}_t, \mathscr{s}(\tilde{{p}}_t, \hat \theta^D, \cdot), \tilde{\mathcal{H}}_t, \hat \theta^D).
\end{equation}

For example, under Bertrand-Nash, the markup function is $\left( \tilde{\mathcal{H}}_t \odot \frac{\partial \mathscr{s}(\tilde{{p}}_t, \hat \theta^D, \cdot)}{\partial {p}_t} \right)^{-1} \mathscr{s}(\tilde{{p}}_t, \hat \theta^D, \cdot)$. The same procedure applies to other models and is limited by the assumptions on conduct. Assumptions on cost efficiencies generated by the merger, quantified in post-merger cost vectors $\tilde{{c}}_{mt}$, can be incorporated by using this vector instead of ${c}_{mt}$ in Equation \eqref{eq:postm_pred}.\\

\noindent\textbf{Discussion:} The framework relies on assumptions on demand and supply. First, the researcher makes assumptions on the demand side to capture how consumers substitute across products. Second, she maintains restrictions on the supply side by specifying firms' cost functions and a model of firm conduct, constant across the pre- and post-merger periods.  We discuss these in turn.

In the context of merger review, demand elasticities may be recovered or calibrated from data on diversion, or estimated from simple demand systems \citep[e.g., logit as in][]{werden1994effects}. In the IO literature, a lot of effort has been devoted to estimating flexible models of demand \citep[e.g., ][]{blp95} that can then be used as input in merger simulation. In what follows, we assume the researcher can either calibrate or estimate substitution. 

On the supply side, we first remark on the scope of the exercise. By assuming that conduct is constant before and after the merger, and it is described by oligopoly equilibrium, the model is only capturing unilateral effects of mergers.\footnote{Merger may have important coordinated effects \citep{p20}. See, e.g., \cite{mw17} for a study that allows for and measures coordinated effects.} Within this framework, reliance on the strong assumptions on conduct and cost, such as Bertrand-Nash pricing and constant marginal cost, may have important consequences for prediction. For instance, \cite{peters2006evaluating} suggests that alternative assumptions on conduct may lead to more reliable conclusions in his setting of airline mergers. This motivates a more flexible approach.

\section{Flexible Model of Supply} \label{sec:flexible-model}

We propose a flexible model of supply that maintains the core economic structure of market equilibrium while relaxing standard assumptions about firm conduct. Our approach recognizes that merger simulation is fundamentally a complex prediction problem with simultaneity -- prices are equilibrium objects that jointly solve a system of demand and supply equations. 

\subsection{Model Setup}

In general, we can express prices in market $t$ as a function of markups and costs:
\begin{equation} \label{eq:general-form}
    p_t = \Delta(s_t,p_t,D(s_t,p_t,x_t);\mathcal{H}_t) + c(s_t,\text{w}_t,\omega_t),
\end{equation}
where $\Delta(\cdot)$ is a markup function and $c(\cdot)$ is a cost function. This formulation nests the standard Bertrand-Nash model but allows for more general forms of strategic interaction. With respect to the framework in Section 4.4 in \cite{bh14}, the expression above unpacks marginal revenue as the sum of price and markup, which entails no loss of generality. We further maintain two assumptions also imposed in \cite{bh14}:

\begin{assumption} \label{assump:uniqueness}
\emph{(Equilibrium Selection)} There exists a unique equilibrium, or the equilibrium selection rule is such that the same $p_t$ arises whenever the vector $(\text{w}_t,x_t,\omega_t,\xi_t)$ is the same.
\end{assumption}

This assumption, similar to Assumption 13 in \cite{bh14}, ensures that observed prices reflect stable equilibrium behavior. 

\begin{assumption} \label{assump:separable}
\emph{(Separability of Cost)} The cost function is separable in unobservable shocks:
\begin{equation} \label{eq:separable}
c(s_t,\text{w}_t,\omega_t) = \tilde{c}(s_t,\text{w}_t) + \omega_t.
\end{equation}
\end{assumption}

This separability assumption follows \cite{bh14} and is essentially without loss of generality, as the unobservable component $\omega_t$ can be defined as the residual between total costs and the component explained by observables.

To identify our model of supply, we will assume that demand is identified. This mirrors the two-step procedure in standard merger simulation, where demand is first estimated or calibrated.

\begin{assumption}\label{assump:knownderiv}
\emph{(Known Demand)} The matrix of demand derivatives is known, so that $D_t = D(s_t,p_t,x_t)$ is observed.
\end{assumption}

We further restrict the class of models we consider, so that prices enter markups only through demand derivatives:

\begin{assumption} \label{assump:markup}
\emph{(Markup Dependence)} The markup function $\Delta$ depends only on market shares $s_t$ and the matrix of demand derivatives $D_t$.
\end{assumption}

Under these assumptions, we can write the price equation as:
\begin{equation} \label{eq:reduced-form}
p_t = h(s_t,D_t,\text{w}_t;\mathcal{H}_t) + \omega_t,
\end{equation}
where $h(\cdot)$ is a flexible function that captures both markup and cost components. Notably, this formulation does not impose separability between markups and costs. We further discuss this assumption below.

\subsection{Discussion} \label{sec:moel-discussion}

Several aspects of our framework deserve emphasis. First, it significantly generalizes the standard merger simulation toolkit while maintaining its essential economic structure. The formulation in Equation \eqref{eq:reduced-form} nests the Bertrand-Nash model as a special case but allows for richer forms of strategic interaction. Assumption \ref{assump:markup} is satisfied by a broad range of conduct models beyond Bertrand-Nash, including Cournot competition, Stackelberg leadership, various forms of partial collusion, and models where firms maximize weighted combinations of profits and consumer surplus.\footnote{For a discussion of how the models above (and more) satisfy this assumption, see Appendix C in \cite{dmqsw24}.}


Second, while our framework is flexible, it is not completely unstructured. The economic content enters through the maintained assumption that prices emerge from equilibrium behavior and through the specification of which variables can enter the pricing function. These restrictions are essential for the identification and interpretation of the results.

There may be more additional economic or statistical restrictions on the $h$ function that a researcher may want to impose. Beyond separability, for instance, one may want $h_j$ to be decreasing in that product's own demand elasticity, as is the case in many standard models. This is in the spirit of the micro-founded economic restrictions that are imposed on nonparametric demand systems in \cite{c20}. It is possible to add these restrictions to our model, and we discuss in Appendix \ref{app:regular} how these can be incorporated within our modeling and estimation framework.

Finally, for merger simulation and other counterfactuals, we can use the estimated function $\hat{h}$ to solve for counterfactual prices that satisfy:
\begin{equation} \label{eq:counterfactual}
\tilde{p}_t = \hat{h}(\mathscr{s}(\tilde{p}_t), D(\tilde{p}_t), \text{w}_t;\tilde{\mathcal{H}}_t) + \hat{\omega}_t,
\end{equation}
where $\tilde{\mathcal{H}}_t$ represents the post-merger ownership structure and $\hat{\omega}_t$ are the estimated residuals. This maintains the equilibrium structure of standard merger simulation while allowing for flexible conduct.

\subsection{Identification}

Equation \eqref{eq:reduced-form} expresses equilibrium prices in a market as a function of endogenous market shares. Therefore, the identification of the flexible supply model must rely on instrumental variables. We specify a moment condition:
\begin{equation} \label{eq:moment}
\mathbb{E}[\omega_{jt} \mid z_{jt}, \text{w}_{jt}] = 0,
\end{equation}
where $z_{jt}$ is a vector of instruments. Following standard practice in the literature, candidate instruments include own and rival product characteristics, rival cost shifters, and taxes. Importantly, the simultaneity of the environment, where the demand equation $s_t = \mathscr{s}(p_t)$ expresses quantities as a function of prices, makes it necessary to maintain exclusion restrictions to identify $h_t$. 

To see why these are necessary, consider the case of simple logit demand where price and characteristics enter demand through a linear index $\delta_{jt} = \alpha p_{jt} + \beta x_{jt} + \xi_{jt}$. Under the \cite{berry1994} inversion, $\delta_{jt} = \log s_{jt} - \log s_{0t}$, so that inverse demand is $p_{jt} = \frac{\log s_{jt} - \log s_{0t}}{\alpha} - \frac{\beta x_{jt}+ \xi_{jt}}{\alpha}$. A concern then is that our flexible strategy would back out this inverse demand function, instead of the supply relation $h_t$. Given that Equation \eqref{eq:reduced-form} is restricted to be a function of $s_t, D_t$, and $\text{w}_t$, an exclusion restriction that prevents our procedure from identifying inverse demand is as follows:

\begin{assumption}\label{assump:exog}
\emph{(Instrument Exogeneity and Exclusion)} The vector of instruments $z_{jt}$ that satisfies $$\mathbb{E}[\omega_{jt} \mid z_{jt}, \text{w}_{jt}] = 0$$ contains demand shifter(s) $x^{(e)}_{jt}$ that are excluded from the vector $\text{w}_{jt}$.
\end{assumption}

Under this exclusion, in the logit example above, the moment condition would be violated for $h=\mathscr{s}^{-1}$ as the implied $\omega_{jt}$ explicitly depends on $x^{(e)}_{jt}$. We also maintain, similar to \cite{bh14}, a completeness assumption on the instruments. 

\begin{assumption}\label{assump:complete}
\emph{(Completeness)} For all functions $B(s_t,D_t,\text{w}_t;\mathcal{H}_t)$ with finite expectation, if \\
$\mathbb{E}[B(s_t,D_t,\text{w}_t;\mathcal{H}_t)\mid z_{jt}, \text{w}_{jt}] = 0$ almost surely, then $B(s_t,D_t,\text{w}_t) = 0$ almost surely.
\end{assumption}

Under these assumptions, we can prove that the function $h_j$ is identified:

\begin{theorem}\label{thm:id}
Under Assumptions \ref{assump:uniqueness}-\ref{assump:complete}, the function $h_j$ is identified for all $j=1,...,J$.    
\end{theorem}

\begin{proof}
The proof follows Theorem 1 in \cite{bh14}. For any $j$, for the true $h_j$ function we have that at all fixed values of $(z_{jt}, \text{w}_{jt})$:
    \begin{align}
        \mathbb{E}[\omega_{jt} \mid z_{jt}, \text{w}_{jt}] &= \mathbb{E}[h_j(s_t,D_t,\text{w}_t;\mathcal{H}_t)\mid z_{jt}, \text{w}_{jt}] - \mathbb{E}[p_{jt}\mid z_{jt}, \text{w}_{jt}]\\
        \mathbb{E}[p_{jt}\mid z_{jt}, \text{w}_{jt}] &= \mathbb{E}[h_j(s_t,D_t,\text{w}_t;\mathcal{H}_t)\mid z_{jt}, \text{w}_{jt}],
    \end{align}
where the first equality follows by Assumption \ref{assump:exog}, and $\mathbb{E}[p_{jt}\mid z_{jt}, \text{w}_{jt}]$ is identified by the data. Suppose that the function $\tilde{h}_j$ is also such that:
$$\mathbb{E}[p_{jt}\mid z_{jt}, \text{w}_{jt}] = \mathbb{E}[\tilde{h}_j(s_t,D_t,\text{w}_t;\mathcal{H}_t)\mid z_{jt}, \text{w}_{jt}],$$
so that $$\mathbb{E}[{h}_j(s_t,D_t,\text{w}_t;\mathcal{H}_t)\mid z_{jt}, \text{w}_{jt}] - \mathbb{E}[\tilde{h}_j(s_t,D_t,\text{w}_t;\mathcal{H}_t)\mid z_{jt}, \text{w}_{jt}] = 0 \qquad a.s.$$
By Assumption \ref{assump:complete}, it follows that the function $B = h_j - \tilde{h}_j = 0$ almost surely, which in turn implies that $h_j$ is identified. Note that Assumption \ref{assump:exog} also rules out that the function $\tilde{h}_j$ could be the inverse demand function.
\end{proof}

It is useful to draw a parallel between our result in Theorem \ref{thm:id} and the identification results in Section 4.4 of \cite{bh14}. While they establish the identification of the cost shocks without imposing a supply model, we show identification of the full supply function. This is because our assumptions go beyond Assumption 7b in \cite{bh14} in several important ways. First, we impose a functional form on marginal revenue in Equation \eqref{eq:general-form}, and unpack marginal revenue as the sum of price and markups. Second, we allow prices to only impact markups through shares and demand derivatives (Assumption \ref{assump:markup}). Third, we assume that demand derivatives are known (Assumption \ref{assump:knownderiv}). These assumptions come at a cost in terms of generality but are fit for our purposes and are satisfied by the standard conduct models considered in most of the empirical literature. In particular, the models in Remark 1 of \cite{bh14} that satisfy their assumptions 7a and 7b, also satisfy our assumptions \ref{assump:uniqueness}-\ref{assump:complete}. 

While Theorem \ref{thm:id} establishes identification of our flexible model, estimation poses significant challenges. Standard nonparametric instrumental variables techniques are likely to perform poorly in finite samples due to the curse of dimensionality. The next section describes how we overcome these challenges using recent advances in machine learning.


\section{Estimation and Inference} \label{sec:vmm}

We now turn to the estimation of the flexible supply function via VMM. To contextualize VMM as our estimator of choice, we first consider Nonparametric Instrumental Variable (NPIV) methods. NPIV estimators are well-documented (e.g., \citealt{chen2007large}; \citealt{carrasco2007linear}) and offer flexibility in modeling complex relationships without restrictive parametric assumptions. Their theoretical foundations provide robust asymptotic properties under appropriate conditions, making them appealing for situations where the functional form of structural equations is unknown.

However, conventional NPIV methods are limited by the curse of dimensionality, which restricts their practical application to low-dimensional settings. In contrast, machine learning-based IV estimation methods, such as VMM and the two neural network-based average derivative estimators presented by \cite{chen2023efficient} address this issue by using parametric models that grow in complexity with sample size. VMM retains conventional NPIV's strengths, such as handling complex variable relationships, while mitigating dimensionality challenges. Notably, VMM coincides with the optimally weighted Generalized Method of Moments (GMM) when the unknown parameter is of fixed, finite dimensionality  (Lemma 1 in \citealt{bennett2023vmm}). This underscores VMM's role as a direct generalization of traditional parametric estimators, adaptable to high-dimensional contexts.\looseness=-1

We now formally explain our VMM-based estimation procedure. Explicitly, we use a moment condition for structural markup:
\begin{equation}
    \E[p_{jt} - h_j({s}_t, {D}_t, w_{t}; \theta, \mathcal{H}_t) \mid z_t, w_t] = 0
\end{equation}
 Given a preliminary consistent estimate $\tilde \theta_N$ estimator solves a min-max program:
\begin{align} \label{eq:vmm}
    \hat \theta_N &= {\rm argmin}_{\theta \in \Theta}{\rm sup}_{f \in \mathcal{F}_N} \frac{1}{TJ} \sum_{j,t} f(z_{jt})' \omega_{jt}(\theta) - \frac{1}{4TJ}\sum_{j,t} (f(z_{jt})' \omega_{jt}( \tilde{\theta}_N))^2 - R_N(f, h)  \\
    & \qquad \text{where } \;\; \omega_{jt}(\theta) = p_{jt} - h_j(s_t, D_t, w_{t}; \theta, \mathcal{H}_t) \nonumber
\end{align}

Following the identification arguments in Section \ref{sec:flexible-model}, the estimator includes a vector of observable cost shifters $w_{jt}$ and a vector of instruments $z_{jt}$.\footnote{In practice, we exclude observable cost shifters from the set of instruments. Given that the function $f$ is highly nonlinear, the cost shifters no longer explicitly instrument for themselves. Simulations reveal better estimates upon the exclusion of $w$ from the set of instruments.} The instruments are necessary to address the endogeneity of market shares $s_t$ and demand derivatives $D_t$. In general, $f \in \calF_N$ and $h \in \calH_N$ belong to sequences of function classes. We use classes of neural networks with growing width and depth, allowing flexible controls of model complexity to cope with the curse of dimensionality. $R_N$ is a regularizer $R_N: \calF_N \times \mathcal{H}_N \to [0,\infty]$ that penalizes the complexity of the neural network; we note in Appendix \ref{app:regular} that regularization can also be used to adhere to economic properties of the supply function. $\theta$ is a potentially large set of parameters that pins down the function $h$. Under standard regularity assumptions, then Theorem 4 in \cite{bennett2023vmm} shows that the estimator in Equation \eqref{eq:vmm} is consistent for the true parameter values $\theta_0$.

\subsection{Inference}

We now turn to the quantification of uncertainty in the flexible supply function. 
Importantly, while VMM offers valid element-wise inference procedures for the underlying nuisance parameters, our focus is on predicting post-merger prices -- a complex functional of these parameters. To address this, we develop an inference method that employs the numerical delta method in conjunction with Holm's step-down procedure and a permutation procedure. This approach enables the quantification of uncertainty in post-merger price predictions through standard errors and confidence intervals, ensuring computational feasibility. If $\tilde \theta_N\stackrel{p}{\to}\theta_0$, under regularity conditions, Theorems (2)-(3) in \cite{bennett2023vmm} imply:
\begin{gather*}
    \sqrt{N}(\hat \theta_N - \theta_0) \stackrel{d}{\to} N(0,\Omega_0^{-1}) \\
    \Omega_0 = \E\left[\E[\nabla_\theta \omega(\theta_0) \mid z, w]' \E[\omega(\theta) \omega(\theta)' \mid z, w]^{-1} \E[\nabla_\theta \omega(\theta_0) \mid z, w] \right]
\end{gather*}

For inference on the supply function $h: \R^b \supset \Theta \to \R^d$ at a set of $d$ observations, the numerical delta method in \cite{hong2015numerical} implies:
\begin{equation} \label{eq:asymptotic}
    \sqrt{N}(h(\hat\theta_N, \mathcal{H}) - h(\theta_0, \mathcal{H})) \stackrel{d}{\to} N(0, \nabla_{\theta'} h(\theta_0, \mathcal{H}) \Omega_0^{-1} \nabla_{\theta'} h(\theta_0, \mathcal{H})')
\end{equation}
Note that $\nabla_{\theta} h(\theta_0; \mathcal{H})$ has dimension $d \times b$. In the simplest case, suppose that $d = 1$ and the parameter of interest is $h_x(\theta_0, \mathcal{H})$. Lemma (9) in \cite{bennett2023vmm} states that, for any $\beta \in \mathbb{R}^b$, we have:
\begin{equation} \label{eq:variance-est}
    \beta^T\Omega_0^{-1} \beta = -\frac{1}{4} \inf_{\gamma\in \mathbb R^b} \sup_{f\in \calF} \E[f(z)' \nabla_\theta \omega(\theta_0) \gamma] - \frac{1}{4} \E[(f(z)' \omega(\theta_0))^2] - 4 \gamma' \beta - R_N(f)
\end{equation}
Taking $\beta = \nabla_\theta h_x(\theta_0, \mathcal{H})$ for $x$, the solution to the optimization problem above yields the asymptotic variance in Equation \eqref{eq:asymptotic}. The gradient $\nabla_\theta h_{x}(\theta_0, \mathcal{H})$ is difficult to compute analytically but numerical differentiation can be employed, e.g., \cite{hong2015numerical}, and $\hat \theta_N$ can be used in place of $\theta_0$. For numerical differentiation to be valid, the $\epsilon$ in the difference has to satisfy $\epsilon\to 0$ and $N\epsilon/\log N\to \infty$ from Theorem 1 in \cite{hong2015numerical}.\footnote{In practice, we use automatic differentiation in \texttt{torch}. Broadly speaking, the module stores a computational graph as the neural network is fit; training requires and stores the gradients in backpropagation. We note that we are still subject to the regularity condition from \cite{hong2015numerical}.}

However, this approach cannot obtain a covariance matrix when $d \geq 2$. We extend the method to provide a simultaneous confidence interval by adapting Holm's Step-Down procedure \citep{holm1979simple} with the estimates for $\hat \sigma_{x_j}^2(\hat \theta, \mathcal{H})$ and $h_{x_j}(\hat{\theta}, \mathcal{H})$ for each $j = 1, ..., d$. The set of critical values $T_{\alpha}$ is known for significance levels $\frac{\alpha}{d + 1 - k}$ with $k = 1, ..., d$. For any ordering of $x$ and fixed ordering $T_{\alpha}$, we can compute the (vectorized) confidence interval $h_x(\hat\theta, \mathcal{H}) \pm N^{-\frac{1}{2}} \hat\sigma_x(\hat \theta, \mathcal{H}) T_{\alpha}$.
To implement the procedure, we construct the intervals for all permutations of $j = 1, ..., d$, resulting in $d!$ permutations of $x$. This is because we must consider any possible ordering of the $p$-values of $x_1, ..., x_d$. The simultaneous confidence interval is subsequently the union of the bounds in each of the $d!$ permutations. Formally, the proposed procedure is as follows:

\begin{algorithm}[H] 
    \caption{Simultaneous Confidence Interval for VMM}
    \SetAlgoNoLine
    \For{each index $j \in \{1, ..., d\} \equiv J$}{
        Estimate $\hat \sigma_{x_j}^2(\hat \theta, \mathcal{H})$ for $\sigma_{x_j}^2(\theta_0, \mathcal{H})$ at $x_j$ by solving Equation \ref{eq:variance-est}
    }
    Fix critical values $T_{\alpha} = \left\{ T_{\alpha_k} : k = 1, ..., d \right\}$ where $\alpha_k = \frac{\alpha}{d + 1 - k}$
    
    \For{each permutation $\tilde{J}$ of indices $J$}{
        Arrange values $\tilde x$ and $\hat\sigma_{\tilde x}$ using permuted indices $\tilde J$
        
        Construct bounds as $h_{\tilde x}(\hat \theta, \mathcal{H}) \pm N^{-\frac{1}{2}} \hat\sigma_{\tilde x}(\hat \theta, \mathcal{H}) T_{\alpha}$ for fixed $T_{\alpha}$
    }
    Collect the simultaneous confidence interval as the union of bounds
    \label{alg:inference}
\end{algorithm}

This novel procedure allows us to quantify uncertainty simultaneously across many predicted prices. Importantly, we can construct confidence intervals on counterfactuals using our method, meaning we can put bounds on predictions for post-merger prices. 


\section{Monte Carlo Simulations} \label{sec:simulation}

We illustrate the predictive performance of our flexible model with simulations. We evaluate the performance of VMM against the standard merger simulation toolkit, where the researcher imposes a potentially misspecified model. To develop this comparison, we first create an artificial merger simulation exercise, where different models are trained/estimated on pre-merger data, and deployed to predict counterfactual post-merger price changes. To dig deeper into where the differences in performance come from, we also (i) evaluate out-of-sample fit of the different models on a subsample of the pre-merger data, and (ii) compute implied pass-through matrices for each model. Both exercises reveal that VMM outperforms other methods because it is able to learn the true structural relationship on the supply side.\\

\noindent \textbf{Data Generation:} We simulate data using the \texttt{pyblp} framework from \cite{cg19} for $T = 100$, $1,000$, and $10,000$ markets. In each market $t$, the number of single-product firms $J_t$ is either $J_t = 2$ or $J_t = 3$ for duopolies and triopolies, respectively.

We adopt a simple logit framework for demand with market size normalized to one. Consumer $i$ receives indirect utility from consuming product $j$ in market $t$, given by:
\begin{equation*}
    u_{ijt} = \alpha p_{jt} + x_{jt} \beta + \xi_{jt} + \varepsilon_{ijt}.
\end{equation*}

We include a vector of a constant and two observed product characteristics in $x_{jt}$ and the price of the product $p_{jt}$. $\xi_{jt}$ and $\varepsilon_{ijt}$ are unobservable shocks at the product-market and individual-product-market-level, respectively. The utility of the outside option is normalized to $u_{i0t} = \varepsilon_{i0t}$. We draw observed product characteristics independently from normal distributions $N(1, 0.25)$, while $\varepsilon_{ijt}$ is assumed to be distributed Type I Extreme Value. The mean taste parameters $\beta$ are set to $\beta = [-4, 3, 6]$ and the price coefficient is set to $\alpha = -0.25$. 

On the supply side, we separately generate data under two conduct models: Bertrand-Nash and profit-weight model. Stacking first-order conditions, the true markups under either model can be represented as:
\begin{equation} \label{eq:sim-price}
    p_t - c_t = \left(-\mathcal{H}_t \odot D'_t(p_t) \right)^{-1} s_t(p_t)
\end{equation}

Under Bertrand-Nash, the ownership matrix is an identity matrix while under profit-weight conduct, we set $\mathcal{H}_{jkt} = 0.75$ for $j \neq k$ and $\mathcal{H}_{jkt} = 1$ for $j = k$. We assume that the marginal cost of producing product $j$ in market $t$ is given by:
\begin{equation*}
    c_{jt} = w_{jt} \gamma + \omega_{jt}
\end{equation*}

We include a constant and two observable cost shifters (excluded from demand) in the vector $w_{jt}$. Marginal cost also depends on a true unobservable cost shock $\omega_{jt}$. We draw observed cost shifters independently from normal distributions $N(1, 0.25)$. The cost parameters are set to $\gamma = [3, 6, 4]$. We adopt the default in \texttt{pyblp} by drawing unobserved demand and cost shocks $\xi_{jt}$ and $\omega_{jt}$ from a mean-zero bivariate normal distribution with variances of one and correlation of 0.9. Under the given model of conduct, prices solve the fixed point of Equation \eqref{eq:sim-price}. The Bertrand-Nash and profit-weight simulations yield mean own-price elasticities $-3.29$ and $-4.32$, and mean diversions to the outside option of $0.34$ and $0.50$, respectively.

In the estimation routine, the sample is split into a training set and a hold-out test set, stratified to keep whole markets intact, prevent overfitting, and evaluate out-of-sample performance. The model is trained on duopoly and triopoly markets. We include both market structures to give the estimator variation in market structure to learn the supply function. The simulation environment mimics a merger simulation exercise in that the estimated model has seen enough variation to extrapolate well to post-merger behavior upon merging two firms in a triopoly market.\\


\noindent \textbf{Merger Simulation Models:} We describe now the set of models that we use to generate predictions, and how each is estimated/fit from the pre-merger data. First, within the standard merger simulation toolkit, we impose three different canonical models indexed by $m$: Bertrand $(m=B)$, monopoly $(m=M)$, and perfect competition (marginal cost pricing, $m=P$). Under each model $m$, we can compute the implied markups $\Delta_m$ and recover marginal cost shocks $\omega^m$ as the residuals to the regression 
$    p_{jt} - \Delta^m_{jt} = w_{jt} \gamma + \omega_{jt}^m$, which we estimate by OLS.
For the VMM model, we estimate the function $h$:
\begin{equation*}\label{eq:flex-vmm}
    p_{jt} = h_j(s_t, D_t, w_{jt}; \mathcal{H}_t) + {\omega}_{jt}
\end{equation*}
with the VMM formulation, obtaining the implied shocks $\hat{\omega}$.



The neural networks are implemented using the \texttt{torch} package in Python. The objective function for the VMM supply function $h_j$ is given by Equation \eqref{eq:vmm}. We do not include regularization in our baseline simulations. The variance estimates use the objective function from Equation \eqref{eq:variance-est}.

When estimating VMM, we consider two aspects of implementation. To illustrate how knowing the demand side can improve performance of the model, we perform estimation with and without the demand derivatives. Given that the demand system is simple logit, neural nets could learn the demand derivatives from market share data alone. However, providing the derivatives directly can improve performance, especially in small samples. We also experiment with different sizes of the neural net: a ``small'' specification that uses a $3 \times 3$ hidden layer, and a ``large'' specification that uses a $100 \times 100$ hidden layer.  



\subsection{Merger Simulation Results}

Now that we have constructed estimators, the ultimate goal of the structural object recovered in estimation is to conduct a merger simulation and compare predictions with the synthetic data. 

To obtain post-merger predictions for model $m$, we find a vector of prices $\tilde{p}^m_t$ in all markets $t$ that simultaneously satisfy the demand system, and the first order conditions given the post-merger ownership matrix $\tilde{\mathcal{H}}_t$. For the standard merger simulation toolkit, we rely on a contraction mapping from \cite{morrow2011contraction}. Using our estimated flexible supply function, we modify this method into a root-finding problem:\footnote{Alternatively, we set up the problem as a contraction wherein we iterate on prices. In practice, we found that root-finding converged market-wise to a solution to a tolerance of $1e-6$ and the solution matched that of the contraction. We opted for the root-finding technique because it was significantly faster.}
\begin{equation} \label{eq:root-prices}
    \tilde{p}_t - \hat{h}(s_t(\tilde{p_t}), D(p_t), w_t; \tilde{\mathcal{H}}_t) - \hat \omega_t = 0
\end{equation}
We rely on the structure of an estimated demand system $s(\cdot)$ and its derivatives $D(\cdot)$ to solve for equilibrium shares and, subsequently, counterfactual prices under post-merger ownership $\tilde{\mathcal{H}}$. We hold all the parameters in demand and cost at their true or estimated values, and hold the cost shocks (for both the true and imposed model) in the post-merger period fixed at their pre-merger values.

\begin{table}[h!]
    \captionsetup{justification=centering}
    \caption{Prediction Error for Bertrand DGP Merger Simulation (Duopolies and Triopolies)}
    \begin{minipage}{0.6\linewidth}
        \centering
        \label{fig:ms_comparison_baseline}
        \includegraphics[width=\textwidth]{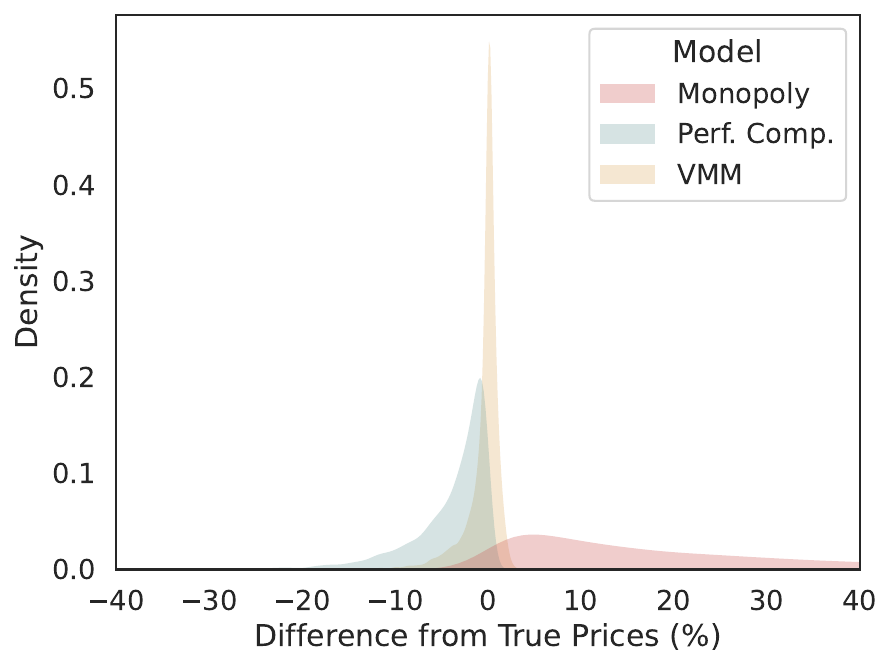}
    \end{minipage}
    \begin{minipage}{0.2\linewidth}
        \centering
        \begin{tabular}{lc}
             Model & MSE \\ \hline 
             Monopoly & 26.26 \\ 
             Perf. Comp. & 1.80 \\ 
             VMM & 0.27
        \end{tabular}
    \end{minipage}
    \begin{tablenotes}[flushleft]
  \setlength\labelsep{0pt}
    \footnotesize
    \item Notes: The figure and table show the distribution and mean squared error (MSE) in prices for the true model (Bertrand), a set of standard models, and the VMM model in the merger simulation exercise. The model is trained on duopolies and triopolies. The MSE is computed on a hold-out test sample restricted to markets where merging firms are present. The neural network used in estimation is small with a $3 \times 3$ hidden layer.
    \end{tablenotes}
    \label{tab:bertrand-du-tri-mergers}
\end{table}


In Table \ref{tab:bertrand-du-tri-mergers}, we compare the distributions of percentage discrepancies $(\tilde{p}^m_{jt} - \tilde{p}_{jt})/\tilde{p}_{jt}$ between predicted and actual post-merger prices. We also compute prediction MSEs for each model, obtained as $\sum_{j,t} (\tilde{p}^m_{jt} - \tilde{p}_{jt})^2$ for each model $m$.

We compare post-merger prices under the true model to prices under different assumptions of conduct to evaluate performance. The monopolist performs the worst with prices far higher than the true prices. Perfect competition significantly underpredicts price changes. Our VMM estimator greatly outperforms all misspecified models and is close to the true prices. It is important to note that the merger of triopolies results in duopolies which the training process is able to see in-sample. In Table \ref{tab:bertrand-tri-mergers}, we train on only triopolies so the merger resulting in duopolies is strictly out of sample. The results are very similar: our estimator greatly outperforms all misspecified models.

\begin{table}[h!]
    \captionsetup{justification=centering}
    \caption{Prediction Error for Profit Weight DGP Merger Simulation (Duopolies and Triopolies)}
    \begin{minipage}{0.6\linewidth}
        \centering
        \label{fig:ms_comparison_complex}
        \includegraphics[width=\textwidth]{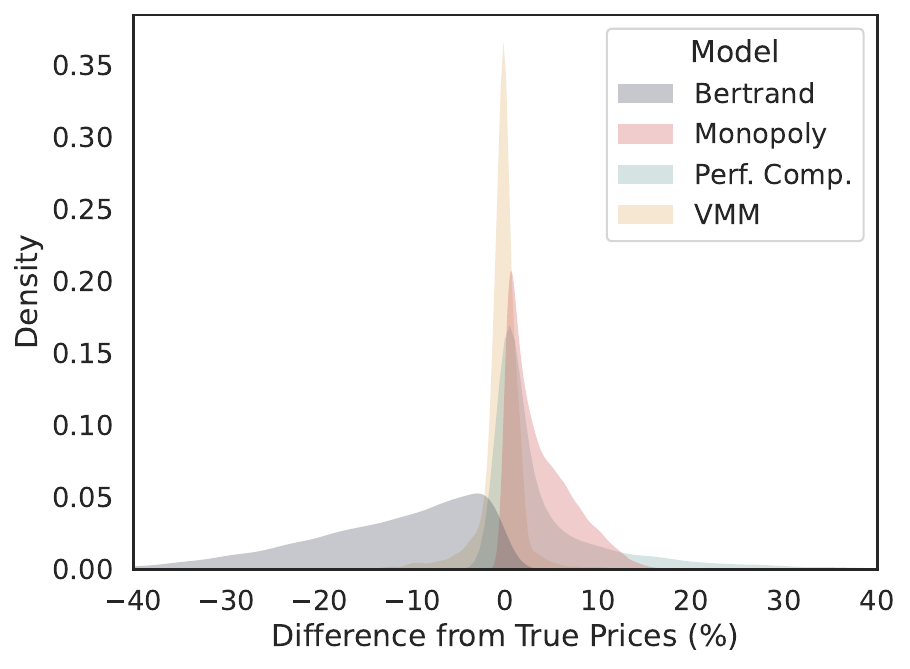}
    \end{minipage}
    \begin{minipage}{0.2\linewidth}
        \centering
        \begin{tabular}{lc}
             Model & MSE \\ \hline 
             Bertrand & 15.19 \\ 
             Monopoly & 1.55 \\ 
             Perf. Comp. & 3.57 \\ 
             VMM & 0.24
        \end{tabular}
    \end{minipage}
    \begin{tablenotes}[flushleft]
  \setlength\labelsep{0pt}
    \footnotesize
    \item Notes: The figure and table show the distribution and mean squared error (MSE) in prices for the true model (profit-weight with $\kappa = 0.75$), a set of standard models, and the VMM model in the merger simulation exercise. The model is trained on duopolies and triopolies. The MSE is computed on a hold-out test sample restricted to markets where merging firms are present. The neural network used in estimation is small with a $3 \times 3$ hidden layer.
    \end{tablenotes}
    \label{tab:profit-weight-du-tri-mergers}
\end{table}

The results in Table \ref{tab:profit-weight-du-tri-mergers} show similar performance in the profit-weight simulations. Bertrand is now misspecified and performs poorly in the merger simulation exercise. Equilibrium prices generated from perfect competition and a monopolist also do not perform well. Our method performs again greatly outperforms all misspecified models and is close to the true post-merger prices. We evaluate performance with a model trained solely on tripolies in Table \ref{tab:profit-weight-tri-mergers} and find that our model again outperforms all others.

\subsection{Simulation Results}

We compute MSE on a hold-out sample to evaluate performance outside the training sample. The results in Table \ref{tab:small-bertrand} and Table \ref{tab:large-profit-weight} vary sample sizes, assumptions on conduct, and the architecture of the estimator.\footnote{Additional results can be found in Table \ref{tab:large-bertrand} and Table \ref{tab:small-profit-weight}.} The first column describes the sample size in terms of the number of markets and the second column denotes whether demand derivatives were included in training. The remaining columns refer to the MSE of the relevant model, i.e., the true model, Bertrand, monopoly, perfect competition, and VMM, respectively.

\input{Tables/small-bertrand-prices-no-nnn} 

Table \ref{tab:small-bertrand} describes the Bertrand simulations with a small neural network for the flexible models of supply. Generally, the shocks implied by perfect competition $\omega^P$ and especially monopoly $\omega^M$ are far away from the true shocks $\omega^B$ and $\omega$. The shocks implied by our VMM estimator $\hat \omega$ outperform all misspecified models. The performance of our estimator without demand derivatives remains stable across sample sizes. 
In Table \ref{tab:large-bertrand}, we increase the size of the neural network and allow more parameters to better fit the nonlinear relationship; we see performance improve with sample sizes even without demand derivatives, fitting the intuition. On the whole, our estimator outperforms all misspecified models, sometimes by an order of magnitude.

\input{Tables/large-profit-weight-prices-no-nnn}

The results in Table \ref{tab:large-profit-weight} are generated with the profit-weight simulations and a large neural network for the flexible models of supply. The collusive profit-weight model is more difficult for a researcher to estimate because there are free parameters in the ownership matrix. A similar pattern emerges in the results compared to the Bertrand data-generating process. Monopoly- and perfect competition-implied shocks, $\omega^M$ and $\omega^P$, respectively, are far from the true shocks, and now Bertrand-implied shocks $\omega^B$ are far from the truth as well. Our VMM estimator outperforms all misspecified models with implied shocks $\hat \omega$ that are close to the ground truth. The model performs better as the sample size increases, especially with the inclusion of demand derivatives. In Table \ref{tab:small-profit-weight}, we include results for a small neural network that cannot capture nonlinearities as well as a large neural network. It remains the case that we outperform all misspecified models even with a small neural network. 

\subsection{Interpretation with Pass-through}

A key concern of using a black box method such as neural networks is the interpretability of the resulting structural supply function $\hat h$. In an attempt to ``peek inside the black box,'' we simulate the pass-through of the resulting supply function and compare it to pass-through under the true model. The pass-through matrix shows us a rough shape of the supply function as we change own and rival costs.

To compute pass-through for the flexible supply function, we increase costs by $10$ percent, loading increases on the residuals $\hat \omega$. We then solve for equilibrium prices under different models of conduct and compare the prices before and after the cost change. Below, we report pass-through matrices for markets with median post-merger inside shares from our different simulation environments.\looseness=-1

\begin{table}[h!]
    \captionsetup{justification=centering}
    \centering
    \caption{Pass-through Comparison: Bertrand vs Profit Weight DGP}
    
    \begin{minipage}[t]{0.5\textwidth}
        \centering
        \textbf{Panel A: Bertrand DGP}
        
        \vspace{0.5em}
        \begin{minipage}[t]{0.45\textwidth}
            \centering
            \captionsetup{labelformat=empty}
            \caption{True Model}
            \begin{tabular}{cc}
                0.60 & 0.10 \\
                0.16 & 0.78
            \end{tabular}
        \end{minipage}%
        \begin{minipage}[t]{0.45\textwidth}
            \centering
            \captionsetup{labelformat=empty}
            \caption{VMM}
            \begin{tabular}{cc}
                0.65 & 0.12 \\
                0.12 & 0.72
            \end{tabular}
        \end{minipage}
    \end{minipage}%
    \begin{minipage}[t]{0.5\textwidth}
        \centering
        \textbf{Panel B: Profit Weight DGP}
        
        \vspace{0.5em}
        \begin{minipage}[t]{0.45\textwidth}
            \centering
            \captionsetup{labelformat=empty}
            \caption{True Model ($\kappa = 0.75$)}
            \begin{tabular}{cc}
                0.98 & -0.00 \\
                -0.43 & 0.41
            \end{tabular}
        \end{minipage}%
        \begin{minipage}[t]{0.45\textwidth}
            \centering
            \captionsetup{labelformat=empty}
            \caption{VMM}
            \begin{tabular}{cc}
                0.96 & -0.01 \\
                -0.44 & 0.46
            \end{tabular}
        \end{minipage}
    \end{minipage}

    \begin{tablenotes}[flushleft]
        \setlength\labelsep{0pt}
        \footnotesize
        \item Notes: The table reports the pass-through behavior of the true models (Bertrand and profit weight with $\kappa = 0.75$) and VMM. Own and rival costs are increased by 10\% (separately) and firms respond in prices under each model, loading cost increases on the model-implied $\omega$ terms. Reported numbers are price changes divided by the relevant cost changes. Market characteristics for Panel A: $c_1 = 15.08, \; c_2 = 15.33, \; s_1 = 0.61, \; s_2 = 0.04$. Market characteristics for Panel B: $c_1 = 15.72, \; c_2 = 10.20, \; s_1 = 0.03, \; s_2 = 0.60$.
    \end{tablenotes}
    \label{tab:pass-through-comparison}
\end{table}

Table \ref{tab:pass-through-comparison} shows pass-through matrices for both Bertrand and profit-weight simulations. Panel A compares the true pass-through matrix in the Bertrand environment with the pass-through implied by VMM, while Panel B presents the same comparison for the profit-weight environment. In both cases, the pass-through implied by VMM is remarkably close to the true pass-through matrix. The flexible model learns markup and cost functions that imply approximately correct pass-throughs. In the Bertrand environment, the model learns to pass-through own and rival cost increases and, in the profit-weight environment, the model learns to internalize some of the rival price increases after cost increases.

\subsection{Inference Performance}

After evaluating predictive performance in sample and in counterfactuals, and the interpretability of the flexible markup and functions, we quantify the uncertainty of the predictions. For ease of exposition, we pick a single observation at which to evaluate our inference procedure.


\input{Tables/inference-small}

The results in Table \ref{tab:inference} present the estimated standard errors and confidence intervals for our predictions under different simulation environments. The first two rows show that our predictions $\hat \psi$ for the Bertrand simulation are close to the true prices $\psi$, even in small samples, and the resulting standard errors and confidence intervals are small. Similarly, the profit-weight environment has quantifiable and tight confidence intervals that are shrinking with increasing sample sizes. 

\section{Empirical Application} \label{sec:application}

As a showcase of our method, we examine the US airline industry. The airline industry has received substantial attention from the IO literature \citep[starting with][]{berry1992estimation, bcs2006} given the rich available data and significant consolidation over the last two decades. The retrospective studies of large mergers have had mixed results \citep{peters2006evaluating}, potentially linked to non-Bertrand conduct \citep[see, e.g., ][for evidence of non-competitive conduct]{ciliberto2014does}. Our method applies well here given the large amount of data and variation in market structure. 

However, we emphasize the illustrative nature of our application. Recent papers have highlighted the dynamic nature of pricing and demand in this market \citep[e.g., ][]{w22,bhow22,gmw24,hnpsw24}, and the role of endogenous network structure \citep{cmt21,lmprsz22,bgr23,yb24}. We abstract away from these important elements to keep our application tied to the standard merger simulation toolkit.\looseness=-1 

\subsection{Background and Data}

The U.S. airline industry has experienced substantial consolidation in the last two decades with legacy carriers and low-cost airlines participating in large mergers. We show a descriptive increase in the Herfindahl-Hirschman Index (HHI) in Figure \ref{fig:hhi}. The earliest merger in our analysis is the Delta-Northwest merger in 2008. The merger was announced on April 14, 2008, and was approved on October 29, 2008, after roughly six months of review by the US Department of Justice (DOJ). Given the limited overlap between the merging airlines' operations, the merger was perceived as having a modest impact on competition. The second merger included is the United-Continental merger in 2010. The DOJ approved the merger after four months of review on August 27, 2010. As a condition of approval, the merged entity was required to lease slots to Southwest at Newark Liberty Airport in New Jersey. Finally, we consider the controversial merger of American Airlines and US Airways. The last of the ``mega-mergers'' that involved two airlines, the deal was announced on November 12, 2013. According to the settlement terms, the merged entity was required to divest slots at several major airports, most prominently at Ronald Reagan Washington National Airport and New York's LaGuardia Airport. More recently, outside our analysis, Alaska Airlines acquired Virgin America and Hawaiian Airlines, and a federal court blocked JetBlue's attempted acquisition of Spirit.\looseness=-1

We construct a database of the US airline industry during the period 2005-2019. We use the 10 percent sample of purchased airline tickets from the well-known Airline Origin and Destination Survey (DB1B) database released by the US Department of Transportation. Following \cite{azar2018anticompetitive} and \cite{kennedy2017competitive}, a market is defined as a pair of cities, regardless of the flight direction. We match cities to Metropolitan Statistical Areas (MSA) and collect data on the populations of these MSAs from the Bureau of Economic Analysis. The geometric mean of endpoint populations is used as a measure of the market size. A product is a one-way trip that services a particular city-pair and is defined at the carrier-market-quarter level. Additional details on the construction of the data can be found in Appendix \ref{sec:data-appendix}.

We evaluate the retrospective unilateral price effects of major airline mergers following previous studies of this industry (e.g., \cite{peters2006evaluating}) and other reduced-form investigations of the price effect of mergers (e.g., \cite{ashenfelter2013price}) by creating control groups of markets not affected by the merger. The rich cross-section of separate markets where airlines compete allows to distinguish between overlap markets -- markets where merging firms both have a sizable presence, which are thus likely to be affected by the mergers -- and markets where there is no overlap between merging airlines. As in \cite{peters2006evaluating}, we define overlap markets as those in which both merging airlines operate flights for the four quarters before the merger. We use the four quarters surrounding the merger approval as our pre- and post-merger samples. The sample is restricted to markets with at most five firms in the pre-merger period. The results, presented in Table \ref{tab:did}, show limited effects of the first two mergers while the American-US merger resulted in higher prices.

\subsection{Demand and Supply Estimation} \label{sec:demand}

\noindent\textbf{Demand Estimation:} We follow \cite{berry2010tracing} in adopting a nested logit demand model. We briefly summarize the model here and provide additional details in Appendix \ref{sec:demand-appendix}. Product characteristics include average fares, the share of nonstop flights, the average distance in thousands of miles, and a squared distance term. We restrict our attention to the major carriers, controlling for the number of fringe firms to capture variation in market structure over time across origin-destination pairs. We include origin-destination fixed effects. Our nesting structure includes all inside goods in one nest. We include instruments to handle endogeneity issues for prices and nests. We use BLP instruments as the average rival distance, the average number of markets a rival serves, and the number of rival carriers.

The results for demand estimation are reported in Table \ref{tab:demand}. In line with the previous literature, we find that consumers prefer a higher share of nonstop flights and incur disutility from more miles traveled. There is strong within-nest substitution. The median own-price elasticity of $-5.17$ matches the literature well.\\

\noindent\textbf{Pre-merger Supply Estimation:} We consider two models of conduct for our supply specifications: (i) Bertrand pricing and (ii) a flexible supply function as described in Section \ref{sec:flexible-model}. In this section, we focus on the flexible specification and relegate details of the Bertrand specification to Appendix \ref{sec:supply-appendix}. 

We estimate the flexible supply model with the VMM technique described in Section \ref{sec:vmm}.\footnote{In our baseline specification, we do not include demand derivatives, although the model should fit them given the data at hand. As a robustness check, we fit a model that includes demand derivatives but found negligible differences in performance. 
} We include in the supply function market shares and the average distance in thousands of miles as an observable cost shifter. We also include origin-destination fixed effects. We instrument the endogenous market shares with BLP instruments formed with the following characteristics: average rival distance, average number of markets a rival serves, and number of rival carriers. Additionally, we include own-product characteristics that do not directly impact marginal costs -- the share of nonstop flights and squared average distance in thousands of miles -- as excluded instruments.

Having estimated the two models, we can evaluate their fit in the pre-merger data. The flexible supply function estimated with VMM significantly outperforms the Bertrand conduct assumption by around 36\%; this margin is similar in the training data and a test sample. We stratify the data by market and split it into 80\% of markets for training, leaving the remaining 20\% of markets as a hold-out test sample to evaluate the fit.

\subsection{Merger Simulation Results}

We examine the merger of American Airlines and US Airways in our counterfactuals. As in the difference-in-differences, we focus on markets with three firms in the pre-merger period and two firms in the post-merger period. We compare the predictions of the models to the true post-merger prices. The observed price differences are presented in Figure \ref{fig:posttruth}.

\begin{figure}[htb]
    \centering
        \caption{Merger Simulation Results}
    \begin{minipage}[t]{0.5\textwidth}
        \centering
        \caption*{Panel A. Simulated Price Changes}
        \includegraphics[width=\textwidth]{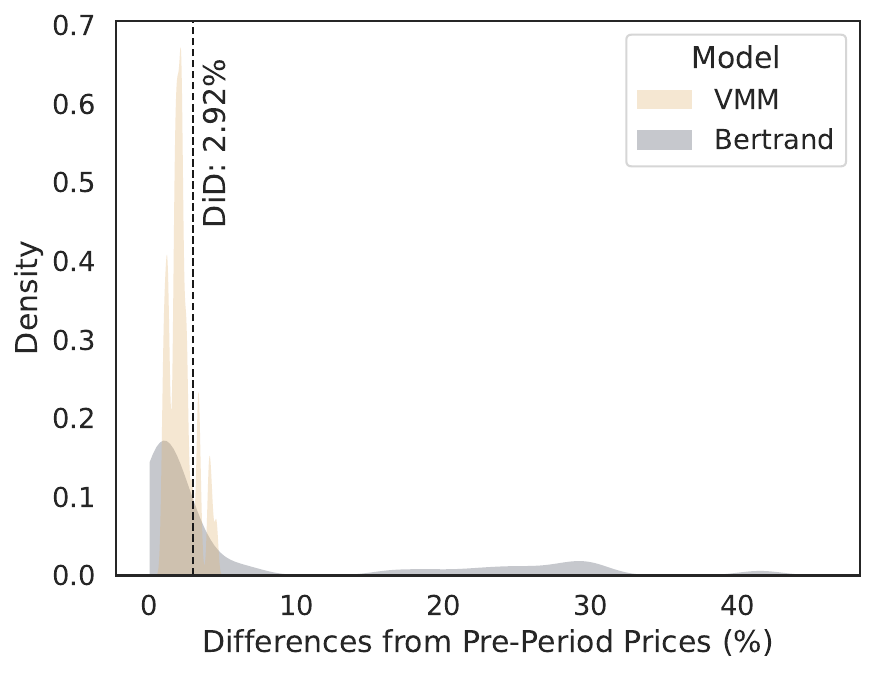}
    \end{minipage}
    \begin{minipage}[t]{0.49\textwidth}
        \centering
        \caption*{Panel B. Post-Merger Price Prediction Error}
        \includegraphics[width=\textwidth]{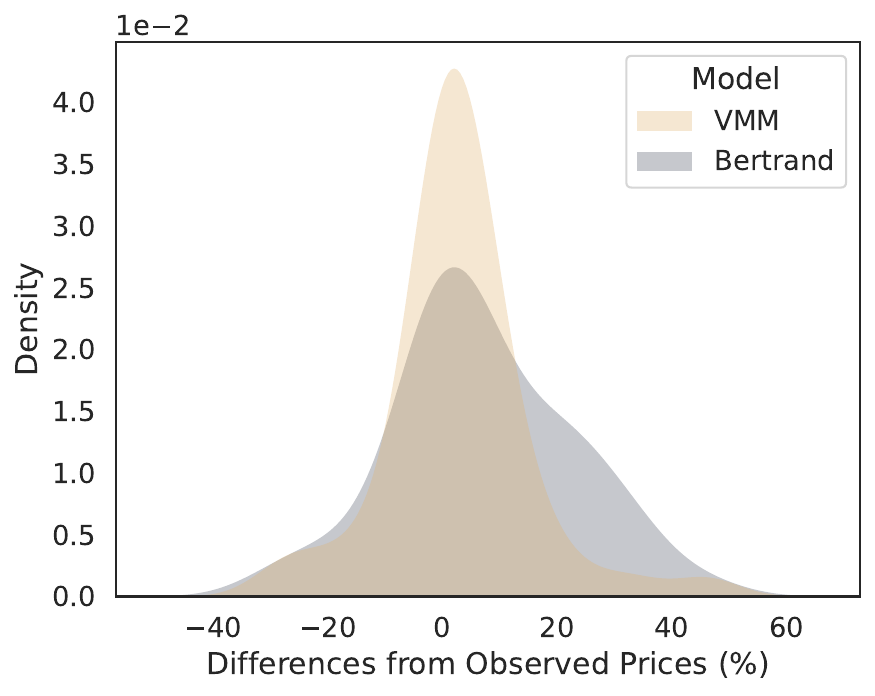}
    \end{minipage}
    \caption*{\footnotesize{Notes: The figure reports merger simulation results for the flexible model estimated with VMM (in yellow) and the standard merger simulation model (in blue). Panel A reports the distribution of percent differences between pre-merger and post-merger predicted prices. Panel B reports the distribution of percent differences between observed and post-merger predicted prices. Figure \ref{fig:post-breakdown} breaks down Panel B by merger status.}}
    \label{fig:msim}
\end{figure}

Figure \ref{fig:msim} presents the results of the merger simulation exercise. Panel A represents the distribution of predicted price increases from the merger under the flexible supply model and under the standard merger simulation with Bertrand. We find that the flexible model predicts a fairly tight distribution of price changes, with an average predicted change of $2.16\%$ (median $2.05\%$). This is in line with the descriptive DiD result, which found a retrospective price change in the data of $2.92\%$ in the markets we examine. In contrast, the standard merger simulation procedure with Bertrand conduct predicts an average price change of $6.66\%$ (median $1.45\%$), which is largely driven by a few large predicted changes (up to $40\%$ increases). The mass of large price increases are predicted for markets that move closer to a monopoly after the merger. More specifically, the market was dominated by US Airways and American Airlines; the third firm was a fringe in the market with a very small market share. After the merger, the market appears close to a monopoly, leading to large price increases under the Bertrand assumption.

To analyze the source of these different predictions, we now compare whether the flexible model or the standard merger simulation model more closely matches realized post-merger prices. Panel B of Figure \ref{fig:msim} plots the distribution of percentage differences between predicted and observed prices for the two methods in the post-merger period. The flexible model estimated with VMM is centered at zero, with a large fraction of predicted prices within $20\%$ of realized prices, and an MSE of $67$. Instead, the standard merger simulation method over-predicts changes and has an overall MSE of $366$. In sum, our flexible model substantially outperforms the standard toolkit when predicting post-merger prices for the American-US Airways merger.\looseness=-1 \\ 

\noindent \textbf{Quantifying Uncertainty:} Finally, we quantify the uncertainty of our predictions in the merger simulation exercise. We follow Algorithm \ref{alg:inference} to construct the confidence intervals. Notably, we use a more conservative Bonferroni correction for ease of exposition, allowing us to present the results with a single set of bounds for each point. We construct bounds for all points in the sample selected for merger simulation.

The results are presented in Appendix Figure \ref{fig:inferencewidth}, where we show the total width of the confidence interval as a summary of uncertainty. The predictions of the inference exercise show that uncertainty is roughly constant with prediction levels. Further, and more importantly, the width is at a reasonably small level (mostly between 5-10 for predicted prices in the range 100-300) in which we can make precise point estimates even with a high-dimensional supply function. 

\section{Conclusion} \label{sec:conclusion}

This paper demonstrates how machine learning methods can enhance merger analysis while maintaining economic structure. Our approach relaxes standard assumptions about firm conduct when rich market data are available, adapting recent advances in variational inference to estimate a flexible model of supply that nests both Bertrand-Nash and alternative models of competition. Applied to the American Airlines-US Airways merger, the method produces more accurate predictions of post-merger prices compared to traditional merger simulation approaches, while maintaining computational tractability. The framework also provides valid statistical inference, allowing antitrust practitioners to quantify uncertainty in their predictions. Beyond mergers, our methodology opens new possibilities for empirical work in IO by showing how machine learning tools can be integrated into structural economic models without sacrificing their economic content.

The method involves important trade-offs and limitations. To mention one, while our method performs well when the post-merger market structure has some precedent in the training data, it may be less reliable for mergers that would create entirely novel market structures, or for which the researcher lacks relatively rich pre-merger data on market outcomes. An important direction for future research is developing hybrid approaches that combine data-driven flexibility with economic restrictions in settings where data are limited. Despite these limitations, our results suggest that careful application of modern machine learning methods can significantly improve our ability to predict merger effects while maintaining the conceptual foundations that make merger simulation such a valuable tool for antitrust analysis.

\pagebreak

\bibliographystyle{ACM-Reference-Format}
\bibliography{MergerSim}

\appendix

\setcounter{table}{0}
\renewcommand{\thetable}{A\arabic{table}}
\setcounter{figure}{0}
\renewcommand{\thefigure}{A\arabic{figure}}
\setcounter{equation}{0}
\renewcommand{\theequation}{A\arabic{equation}}

\pagebreak

\section{Additional Simulation Results} \label{sec:simulation-appendix}

\input{Tables/large-bertrand-prices-no-nnn}
\input{Tables/small-profit-weight-prices-no-nnn}

\begin{table}[H]
    \captionsetup{justification=centering}
    \caption{Prediction Error for Bertrand DGP \\ Merger Simulation (Triopolies)}
    \begin{minipage}{0.6\linewidth}
        \captionsetup{justification=centering}
        \centering
        \label{fig:ms_comparison_baseline_tri}
        \includegraphics[width=\textwidth]{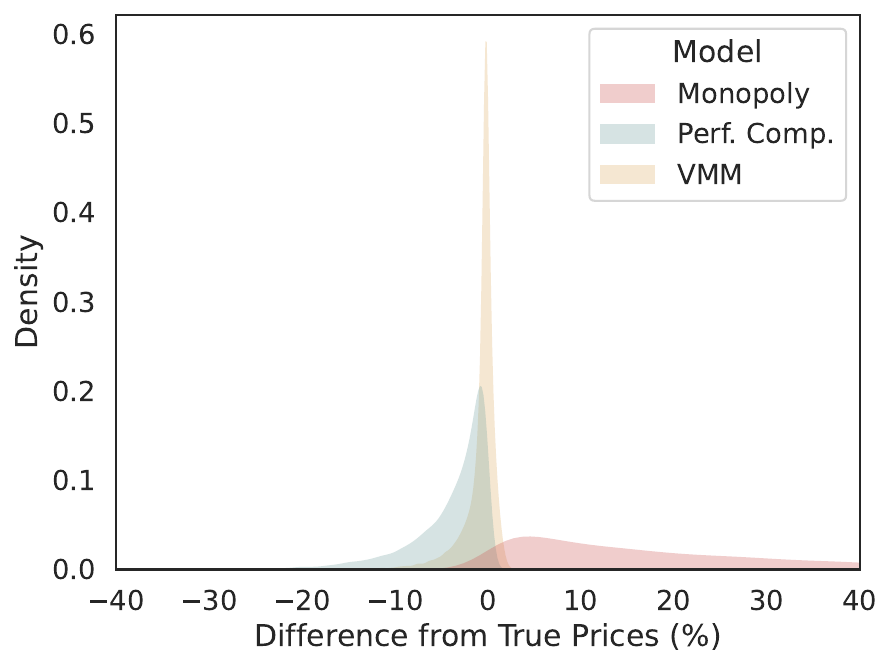}
    \end{minipage}
    \begin{minipage}{0.2\linewidth}
        \centering
        \begin{tabular}{lc}
             Model & MSE \\ \hline 
             Bertrand & 0.00 \\ 
             Monopoly & 25.94 \\ 
             Perf. Comp. & 1.81 \\ 
             VMM & 0.28
        \end{tabular}
    \end{minipage}
    \begin{tablenotes}[flushleft]
  \setlength\labelsep{0pt}
    \footnotesize
    \item Notes: The figure and table show the distribution and mean squared error (MSE) in prices for the true model (Bertrand), a set of standard models, and the VMM model in the merger simulation exercise. The model is trained on triopolies. The MSE is computed on a hold-out test sample restricted to markets where merging firms are present. The neural network used in estimation is small with a $3 \times 3$ hidden layer.
    \end{tablenotes}
    \label{tab:bertrand-tri-mergers}
\end{table}

\begin{table}[H]
    \captionsetup{justification=centering}
    \caption{Prediction Error for Profit Weight DGP \\ Merger Simulation (Triopolies)}
    \begin{minipage}{0.6\linewidth}
        \centering
        \label{fig:ms_comparison_complex_tri}
        \includegraphics[width=\textwidth]{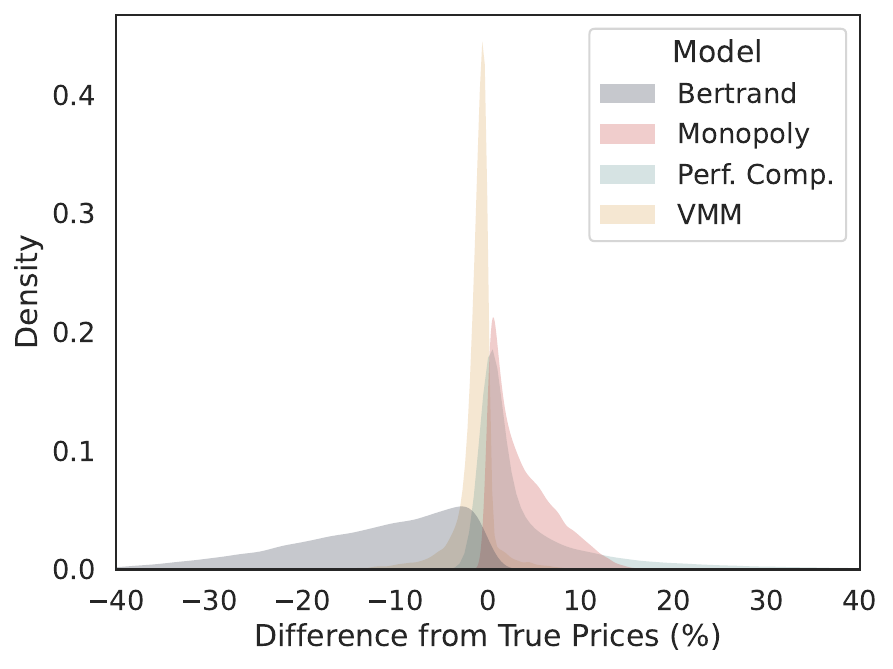}
    \end{minipage}
    \begin{minipage}{0.2\linewidth}
        \centering
        \begin{tabular}{lc}
             Model & MSE \\ \hline 
             Bertrand & 15.02 \\ 
             Monopoly & 1.53 \\ 
             Perf. Comp. & 3.54 \\ 
             VMM & 0.28
        \end{tabular}
    \end{minipage}
    \begin{tablenotes}[flushleft]
  \setlength\labelsep{0pt}
    \footnotesize
    \item Notes: The figure and table show the distribution and mean squared error (MSE) in prices for the true model (Bertrand), a set of standard models, and the VMM model in the merger simulation exercise. The model is trained on triopolies. The MSE is computed on a hold-out test sample restricted to markets where merging firms are present. The neural network used in estimation is small with a $3 \times 3$ hidden layer.
    \end{tablenotes}
    \label{tab:profit-weight-tri-mergers}
\end{table}

\setcounter{table}{0}
\renewcommand{\thetable}{B\arabic{table}}
\setcounter{figure}{0}
\renewcommand{\thefigure}{B\arabic{figure}}
\setcounter{equation}{0}
\renewcommand{\theequation}{B\arabic{equation}}

\section{Economic and Statistical Restrictions}\label{app:regular}

As noted in the main text, one may want to impose additional economic or statistical restrictions on the $h$ function. For instance, one may want $h_j$ to be decreasing in that product's own demand elasticity, as is the case in many standard models. This is in the spirit of the micro-founded economic restrictions that are imposed on nonparametric demand systems in \cite{c20}. It is possible to add these restrictions to our model through regularization, restrictions on weights and activation functions, neural network architecture, or some combination of these. We discuss each of these in turn within the example of monotonicity in own demand elasticities.

The first approach is to incorporate additional components in the regularization term $R_N$ in Equation \eqref{eq:vmm}. Define a set of $n$ own-demand elasticities in increasing order as $D = (D_{(1)}, ..., D_{(n)})$. An example of a simple regularization term $R^M$ for monotonicity is the following:
\begin{equation*}
    R^M(h) = \sum_{i = 2}^n (\max\{h(D_{(i)}) -  h(D_{(i - 1)}), 0 \})^2
\end{equation*}

We condition on $s_t$, $w_t$, $\theta$, and $\mathcal{H}_t$ in the supply function $h$, suppressing them for notational simplicity. If $h(D_{(i)}) \leq h(D_{(i-1)})$, there is no additional penalty on the supply function, but otherwise, we penalize the squared first difference of own-demand elasticities in the spirit of a ridge regression. We note that the choice of regularization is a degree of freedom for the researcher.

The last two approaches are related to the rich computer science literature on monotonic neural networks, starting with \cite{sill1997monotonic}. The first approach enforces constraints on weights and activation functions in particular neural network layers, e.g., \cite{you2017deep}. The user can specify inputs, such as own-demand elasticities, in which the output is monotonic. A second approach, detailed in \cite{wehenkel2019unconstrained}, uses the architecture of the neural network to enforce a constant sign of the derivative of the approximated function without imposing additional constraints. In our specific example, we can restrict the derivative of $h$ with respect to own-demand elasticity to be negative.

The discussion above focuses on a single example of an economic restriction. These approaches can be adapted and combined to impose additional economic or statistical restrictions on the supply function.

\setcounter{table}{0}
\renewcommand{\thetable}{C\arabic{table}}
\setcounter{figure}{0}
\renewcommand{\thefigure}{C\arabic{figure}}
\setcounter{equation}{0}
\renewcommand{\theequation}{C\arabic{equation}}

\section{Additional Empirical Results} 

\subsection{Data Construction} \label{sec:data-appendix}

We construct a database of the US airline industry from 2005-2019. We obtained a quarterly random 10 percent sample of purchased airline tickets from the well-known Airline Origin and Destination Survey (DB1B) database released by the US Department of Transportation. Following \cite{azar2018anticompetitive} and \cite{kennedy2017competitive}, a market is defined as a pair of cities, regardless of the flight direction. We match cities to Metropolitan Statistical Areas and collect data on the populations of these MSAs from the Bureau of Economic Analysis. A product is a one-way trip that services a particular city-pair and is defined at the carrier-market-quarter level. Market sizes are measured as the geometric mean of the origin-destination endpoint populations

\subsubsection{Sample Selection}

We exclude markets with fewer than 20 passengers per day, as airline behavior on these thin, possibly seasonal, routes is unlikely to represent normal competitive behavior in the industry. We also drop itineraries with a ticket carrier change at the connecting airport since these tickets cannot be assigned to a unique ticketing carrier. Finally, we drop every ticket with a fare lower than \$25 and higher than \$2,500 since these tickets are likely the result of reporting errors.

For each carrier-market-quarter, we begin by calculating the product's average price, total passengers, and average distance. Additionally, we construct each product's extra miles - the difference between the average distance in miles and the nonstop distance in the market - and the fraction of nonstop tickets sold. Averages are weighted by the number of passengers. We remove any products with fewer than 800 quarterly passengers. These products arguably have a weak impact on the competitive behavior of carriers with higher market shares, and although this is standard practice in this literature (e.g. \cite{berry2010tracing}) because the dimensionality of the input space is very important in our application, this rule allows us to find more markets with effectively fewer carriers. Summary statistics for the analysis sample are presented in Table \ref{tab:summary}.

\begin{figure}[h!]
    \centering
        \caption{Concentration in the Airline Industry}
    \includegraphics[width=0.7\textwidth]{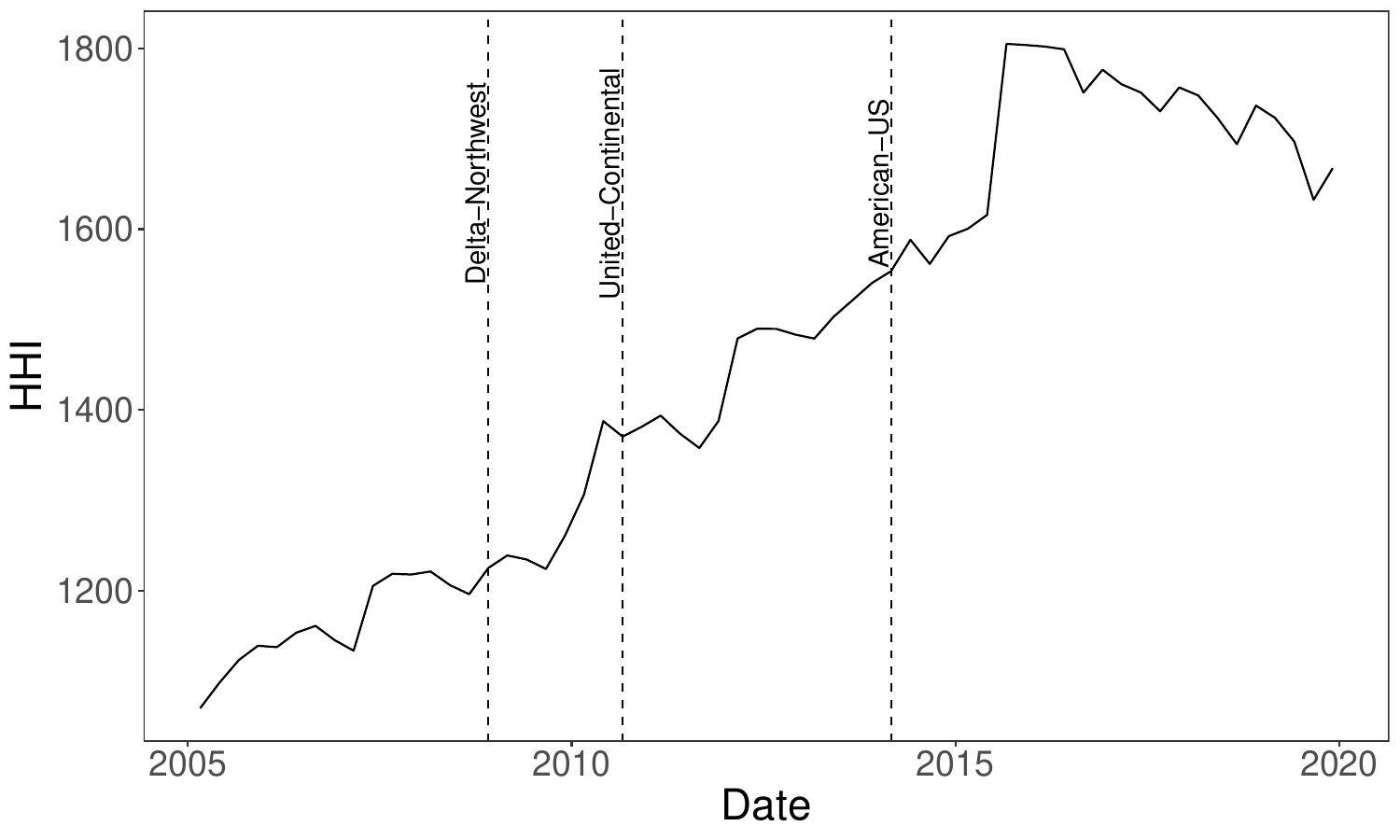}
    \caption*{Notes: The figure plots the evolution of the national HHI of the airline industry during the period 2005-2019. We use passenger counts to construct market shares used in the calculation of HHI.}
    \label{fig:hhi}
\end{figure}

\input{Tables/summary}

\subsection{Demand Estimation} \label{sec:demand-appendix}

We include additional details on demand estimation introduced in Section \ref{sec:demand}. In our demand system, consumer $i$ receives utility from product $j$ in market $t$ with the following indirect utility:
\begin{equation*}
    u_{ijt} = \alpha p_{jt} + x_{jt} \beta + \xi_{m(t)} + \xi_{jt} + \zeta_{it} + (1 - \rho) \varepsilon_{ijt}
\end{equation*}

The vector $x$ includes the share of nonstop flights, average distance in thousands of miles, the squared term of average distance in thousands of miles, and the logged number of fringe firms (plus one to avoid zero issues). The last term is included to focus on the demand for major carriers while controlling for additional variation over time in market structure across origin-destination pairs. The term $\xi_{t}$ is a set of origin-destination fixed effects. $\xi_{jt}$ and $\zeta_{it} + \varepsilon_{ijt}$ are unobservable shocks at the product-market and individual-product-market levels, respectively. We assume that $\varepsilon_{ijt}$ is distributed Type I Extreme Value and $\zeta_{it}$ is distributed according to the conjugate distribution. We close the model by normalizing the utility of consumer $i$ from the outside option to $u_{i0t} = \varepsilon_{i0t}$. Given the structure of utility and distributional assumptions, market shares $s_{jt}$ are a function of observables, unobservables, and parameters in the standard form from \cite{blp95}.

The identifying assumption for demand is that the moment condition $\mathbb{E}[\xi_{jt} z_{jt}^D] = 0$ holds for a vector of demand instruments $z_{jt}^D$. Following \cite{blp95}, we include the average rival distance, the average number of markets a rival serves, and the number of rival carriers. The last instrument is especially useful for identifying the nesting parameter. The results of demand estimation are presented in Table \ref{tab:demand}. The results and median own-price elasticities are in line with the literature.

\subsection{Supply Estimation} \label{sec:supply-appendix}

In Section \ref{sec:demand}, we briefly introduced the supply side and focused on the flexible supply function. In the Bertrand specification, we assume that inferred marginal costs are linear in observable cost shifters:
\begin{equation*}
    p_{jt} - \Delta^B_{jt} \equiv c_{jt} = w_{jt}\gamma + \Gamma_{m(j)} + \omega_{jt}
\end{equation*}

We include only the average distance in thousands of miles in $w_{jt}$ and origin-destination fixed effects. The results are presented in Table \ref{tab:supply}. We find that distance is positively associated with marginal costs inferred under the Bertrand assumption of conduct.

\clearpage
\subsection{Baseline Empirical Results}

\input{Tables/did}

\input{Tables/demand}
\input{Tables/supply}

\subsection{Counterfactuals} \label{sec:counterfactual-appendix}

\begin{figure}[h!]
    \centering
        \caption{Price Change Distribution}
    \includegraphics[width=0.75\textwidth]{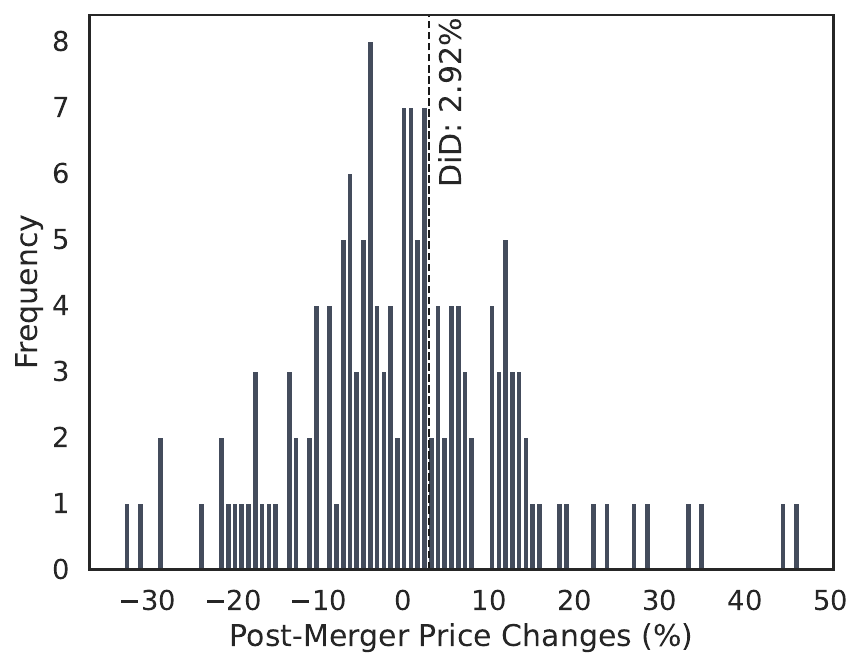}
    \caption*{Notes: The figure plots the observed post-merger price changes after the US-AA merger.}
    \label{fig:posttruth}
\end{figure}

\begin{figure}[h!]
    \centering
        \caption{Post-Merger Price Prediction Error by Status}
    \includegraphics[width=\textwidth]{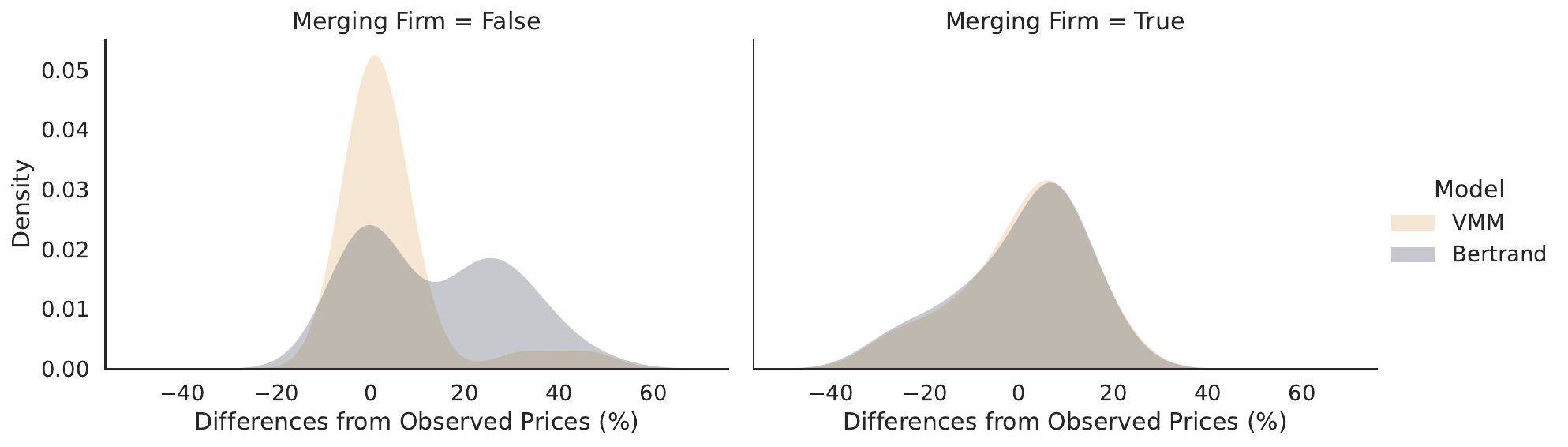}
    \caption*{Notes: The figure reports merger simulation results for the flexible model estimated with VMM (in yellow) and the standard merger simulation model (in blue) broken down by firm status. It reports the distribution of percent differences between observed and post-merger predicted prices.}
    \label{fig:post-breakdown}
\end{figure}

\begin{figure}[h!]
    \caption{Width of Confidence Intervals}
    \centering
    \includegraphics[width=0.7\textwidth]{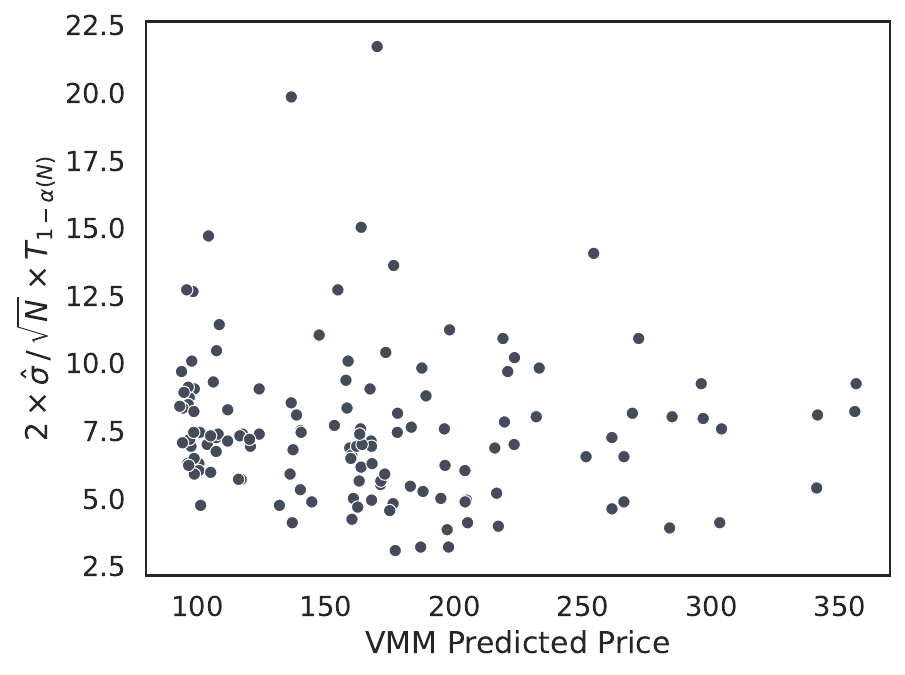}
    \caption*{\footnotesize{Notes:} The figure represents the scatter plot of VMM predicted prices (on the horizontal axis) against the width of the confidence interval for the predictions. Variance is estimated for all firms in markets affected by the US Airways-American merger.}
    \label{fig:inferencewidth}
\end{figure}

\end{document}

%% file: Tables/small-bertrand-prices-no-nnn.tex
\begin{table}[htb]
\footnotesize
\caption{MSE Across Models (Bertrand DGP, Small Network)}
\label{tab:small-bertrand}
\centering
\begin{threeparttable}
        \begin{widetable}{.98\columnwidth}{lcccccc}\toprule
             $T$  & True Model & \multicolumn{3}{c}{Standard Models} & VMM & $D_t$ included \\ \cline{3-5}
             & &  $B$ & $M$ & $P$ & & \\
             \midrule
             \multirow{2}{*}{$100$}  & \multirow{2}{*}{$0.895$} & \multirow{2}{*}{$0.895$} & \multirow{2}{*}{$578.773$} & \multirow{2}{*}{$9.411$} & $1.910$ & No \\
               &  &  & &  & $1.708$ & Yes \\ 
             \midrule
             \multirow{2}{*}{$1,000$}  & \multirow{2}{*}{$0.889$} & \multirow{2}{*}{$0.889$} & \multirow{2}{*}{$1022.812$} & \multirow{2}{*}{$8.170$} & $2.211$ & No \\
              &  &  &  &  & $1.065$  & Yes \\ 
             \midrule
             \multirow{2}{*}{$10,000$}  & \multirow{2}{*}{$0.962$} &  \multirow{2}{*}{$0.962$} & \multirow{2}{*}{$1788.816$} & \multirow{2}{*}{$8.487$} & $2.355$ & No \\
               &  &  &  &  & $1.146$ & Yes \\ \bottomrule
        \end{widetable} 
        \begin{tablenotes}[flushleft]
  \setlength\labelsep{0pt}
    \footnotesize
    \item Notes: The table reports the mean squared error (MSE) in prices for the true model (Bertrand), a set of standard models, and the VMM model. The model is trained on duopolies and triopolies. The MSE is computed on a hold-out test sample restricted to markets where merging firms are present. The neural network used is small with a $3 \times 3$ hidden layer.
    \end{tablenotes}
    \end{threeparttable}
\end{table}

%% file: Tables/large-profit-weight-prices-no-nnn.tex
\begin{table}[htb]
\footnotesize
\caption{MSE Across Models (Profit Weight DGP, Large Network)}
\label{tab:large-profit-weight}
\centering
\begin{threeparttable}
        \begin{widetable}{.98\columnwidth}{lcccccc}\toprule
             $T$  & True Model & \multicolumn{3}{c}{Standard Models} & VMM & $D_t$ included \\ \cline{3-5}
             & &  $B$ & $M$ & $P$ & & \\
             \midrule
             \multirow{2}{*}{$100$}  & \multirow{2}{*}{$0.895$} & \multirow{2}{*}{$8.351$} & \multirow{2}{*}{$5.973$} & \multirow{2}{*}{$13.372$} & $1.845$ & No \\
               &  &  & &  & $2.816$ & Yes \\ 
             \midrule
             \multirow{2}{*}{$1,000$}  & \multirow{2}{*}{$0.889$} & \multirow{2}{*}{$7.877$} & \multirow{2}{*}{$6.650$} & \multirow{2}{*}{$9.962$} & $1.673$ & No \\
              &  &  &  &  & $1.297$ & Yes \\ 
             \midrule
             \multirow{2}{*}{$10,000$}  & \multirow{2}{*}{$0.962$} &  \multirow{2}{*}{$8.719$} & \multirow{2}{*}{$7.049$} & \multirow{2}{*}{$10.853$} & $1.175$ & No \\
               &  &  &  &  & $1.177$ & Yes \\ \bottomrule
        \end{widetable} 
        \begin{tablenotes}[flushleft]
  \setlength\labelsep{0pt}
    \footnotesize
    \item Notes: The table reports the mean squared error (MSE) in prices for the true model (profit weight with $\kappa = 0.75$), a set of standard models, and the VMM model. The model is trained on duopolies and triopolies. The MSE is computed on a hold-out test sample restricted to markets where merging firms are present. The neural network used in estimation is large with a $100 \times 100$ hidden layer.
    \end{tablenotes}
    \end{threeparttable}
\end{table}

%% file: Tables/inference-small.tex
\begin{table}[h!]
    \centering 
    \captionsetup{font=normalsize}
    \caption{Inference Comparison by Sample Size (Small Network)} 
    \begin{threeparttable}
    \begin{tabular}{llcccc} 
    \\[-1.8ex] 
    \hline \\[-1.8ex] 
 Model & Sample Size & $\psi$ & $\hat{\psi}$ & $\hat{\sigma} / \sqrt{N}$ & Interval \\ \hline \\[-1.8ex]
 Bertrand & N = 253 & 21.014 & 21.092 & 0.817 & [18.673, 23.512] \\
 Bertrand & N = 2,579 & 20.341 & 20.474 & 0.057 & [20.305, 20.642] \\
 Profit Weight & N = 253 & 17.321 & 12.907 & 0.309 & [11.991, 13.822] \\
 Profit Weight & N = 2,579 & 17.375 & 15.554 & 0.099 & [15.261, 15.847] \\ [1ex] \hline
    \end{tabular} 
    \begin{tablenotes}[flushleft]
  \setlength\labelsep{0pt}
    \footnotesize
    \item Notes: The table reports the VMM variance estimates for a single observation in the data under different models and sample sizes. The model is trained on duopolies and triopolies. The true price $\psi$ and the model-implied price $\hat \psi$. The standard error $\hat \sigma / \sqrt{N}$ is reported with the implied confidence interval. The neural network used in estimation is small with a $3 \times 3$ hidden layer.
    \end{tablenotes}
    \end{threeparttable}
    \label{tab:inference}
\end{table}

%% file: Tables/large-bertrand-prices-no-nnn.tex
\begin{table}[htb]
\footnotesize
\caption{MSE Across Models (Bertrand DGP, Large Network)}
\label{tab:large-bertrand}
\centering
\begin{threeparttable}
        \begin{widetable}{.98\columnwidth}{lcccccc}\toprule
             $T$  & True Model & \multicolumn{3}{c}{Standard Models} & VMM & $D_t$ included \\ \cline{3-5}
             & &  $B$ & $M$ & $P$ & & \\
             \midrule
             \multirow{2}{*}{$100$}  & \multirow{2}{*}{$0.895$} & \multirow{2}{*}{$0.895$} & \multirow{2}{*}{$578.773$} & \multirow{2}{*}{$9.411$} & $2.083$ & No \\
               &  &  & &  & $1.126$ & Yes \\ 
             \midrule
             \multirow{2}{*}{$1,000$}  & \multirow{2}{*}{$0.889$} & \multirow{2}{*}{$0.889$} & \multirow{2}{*}{$1022.812$} & \multirow{2}{*}{$8.170$} & $1.091$ & No \\
              &  &  &  &  & $1.103$ & Yes \\ 
             \midrule
             \multirow{2}{*}{$10,000$}  & \multirow{2}{*}{$0.962$} &  \multirow{2}{*}{$0.962$} & \multirow{2}{*}{$1788.816$} & \multirow{2}{*}{$8.487$} & $1.226$ & No \\
               &  &  &  &  & $1.385$ & Yes \\ \bottomrule
        \end{widetable} 
        \begin{tablenotes}[flushleft]
  \setlength\labelsep{0pt}
    \footnotesize
    \item Notes: The table reports the mean squared error (MSE) in prices for the true model (Bertrand), a set of standard models, and the VMM model. The model is trained on duopolies and triopolies. The MSE is computed on a hold-out test sample restricted to markets where merging firms are present. The neural network used in estimation is large with a $100 \times 100$ hidden layer.
    \end{tablenotes}
    \end{threeparttable}
\end{table}

%% file: Tables/small-profit-weight-prices-no-nnn.tex
\begin{table}[htb]
\footnotesize
\caption{MSE Across Models (Profit Weight DGP, Small Network)}
\label{tab:small-profit-weight}
\centering
\begin{threeparttable}
        \begin{widetable}{.98\columnwidth}{lcccccc}\toprule
             $T$  & True Model & \multicolumn{3}{c}{Standard Models} & VMM & $D_t$ included \\ \cline{3-5}
             & &  $B$ & $M$ & $P$ & & \\
             \midrule
             \multirow{2}{*}{$100$}  & \multirow{2}{*}{$0.895$} & \multirow{2}{*}{$8.351$} & \multirow{2}{*}{$5.973$} & \multirow{2}{*}{$13.372$} & $3.355$ & No \\
               &  &  & &  & $3.694$ & Yes \\ 
             \midrule
             \multirow{2}{*}{$1,000$}  & \multirow{2}{*}{$0.889$} & \multirow{2}{*}{$7.877$} & \multirow{2}{*}{$6.650$} & \multirow{2}{*}{$9.962$} & $3.120$ & No \\
              &  &  &  &  & $2.006$ & Yes \\ 
             \midrule
             \multirow{2}{*}{$10,000$}  & \multirow{2}{*}{$0.962$} &  \multirow{2}{*}{$8.719$} & \multirow{2}{*}{$7.049$} & \multirow{2}{*}{$10.853$} & $2.993$ & No \\
               &  &  &  &  & $2.170$ & Yes \\ \bottomrule
        \end{widetable} 
        \begin{tablenotes}[flushleft]
  \setlength\labelsep{0pt}
    \footnotesize
    \item Notes: The table reports the mean squared error (MSE) in prices for the true model (profit weight with $\kappa = 0.75$), a set of standard models, and the VMM model. The model is trained on duopolies and triopolies. The MSE is computed on a hold-out test sample restricted to markets where merging firms are present. The neural network used in estimation is small with a $3 \times 3$ hidden layer.
    \end{tablenotes}
    \end{threeparttable}
\end{table}

%% file: Tables/summary.tex
\begin{table}[H] \centering 
\footnotesize
\begin{minipage}{\textwidth}
  \caption{Summary Statistics} 
  \label{tab:summary} 
\begin{tabular}{@{\extracolsep{5pt}}lccccccc} 
\\[-1.8ex]\hline 
\hline \\[-1.8ex] 
Statistic & \multicolumn{1}{c}{Mean} & \multicolumn{1}{c}{St. Dev.} & \multicolumn{1}{c}{Min} & \multicolumn{1}{c}{Pctl(25)} & \multicolumn{1}{c}{Median} & \multicolumn{1}{c}{Pctl(75)} & \multicolumn{1}{c}{Max} \\ 
\hline \\[-1.8ex] 
Average Fare & 216.782 & 82.819 & 25.000 & 169.774 & 209.072 & 252.704 & 2,492.017 \\ 
Total Passengers & 352.395 & 1,184.007 & 1 & 19 & 69 & 210 & 70,909 \\ 
Average Distance & 1,386.534 & 688.694 & 67.000 & 861.000 & 1,255.081 & 1,872.045 & 7,731.500 \\ 
Average Nonstop Miles & 1,170.974 & 618.988 & 67.000 & 678.000 & 1,035.000 & 1,600.000 & 2,783.000 \\ 
Average Extra Miles & 215.560 & 247.441 & $-$1.000 & 37.548 & 136.000 & 305.862 & 5,118.000 \\ 
Share Nonstop & 0.205 & 0.363 & 0.000 & 0.000 & 0.000 & 0.179 & 1.000 \\ 
Origin Hub & 0.153 & 0.360 & 0 & 0 & 0 & 0 & 1 \\ 
Dest. Vacation & 0.096 & 0.294 & 0 & 0 & 0 & 0 & 1 \\ 
LCC & 0.207 & 0.405 & 0 & 0 & 0 & 0 & 1 \\ 
Major & 0.920 & 0.271 & 0 & 1 & 1 & 1 & 1 \\ 
Legacy & 0.733 & 0.443 & 0 & 0 & 1 & 1 & 1 \\ 
Presence & 0.148 & 0.142 & 0.00000 & 0.050 & 0.099 & 0.201 & 1.000 \\ 
Num Markets & 50.434 & 30.156 & 1 & 27 & 47 & 72 & 146 \\ 
Share & 0.001 & 0.003 & 0.00000 & 0.0001 & 0.0004 & 0.001 & 0.095 \\ 
Within Share & 0.184 & 0.219 & 0.00001 & 0.022 & 0.096 & 0.268 & 1.000 \\ 
\hline \\[-1.8ex] 
\end{tabular} 
\begin{tablenotes}[flushleft]
        \setlength\labelsep{0pt}
        \footnotesize
        \item Notes: The table reports summary statistics for our analysis sample. We include the mean, standard deviation, 25th percentile, median, 75th percentile, and max of each of the variables included in our sample. 
    \end{tablenotes}
\end{minipage}
\end{table}

%% file: Tables/did.tex
\begin{table}[H] \centering 
  \caption{Difference-in-Differences Estimates} 
  \label{tab:did} 
\begin{threeparttable}
        \begin{widetable}{.98\columnwidth}{lccc}\toprule
    & \multicolumn{3}{c}{$\log$(Fare)}\\
                                    & DL-NW         & UA-CO          & AA-US \\   
                                    & (1)           & (2)            & (3)\\  
   \midrule 
   Treated $\times$ Post            & -0.0322                 & 0.0103                  & 0.0292$^{***}$\\   
                                    & (0.0229)                & (0.0165)                & (0.0091)\\   
   Share Nonstop                    & -0.2753$^{***}$         & -0.2620$^{***}$         & -0.1160$^{***}$\\   
                                    & (0.0416)                & (0.0513)                & (0.0351)\\   
   Average Distance                 & $5.66\times 10^{-5}$    & $7.8\times 10^{-5}$     & 0.0002$^{***}$\\   
                                    & ($5.83\times 10^{-5}$)  & ($6.88\times 10^{-5}$)  & ($3.18\times 10^{-5}$)\\    
    \\
   $R^2$                            & 0.5817                 & 0.4897                 & 0.4001\\  
   Observations                     & 3,938                   & 5,178                   & 9,965\\  
   Treated (\%)                     & 19.880                  & 22.210                  & 52.230
    \\ \midrule
   Origin-destination fixed effects & $\checkmark$    & $\checkmark$    & $\checkmark$\\   
   Year-quarter fixed effects       & $\checkmark$    & $\checkmark$    & $\checkmark$\\
 \bottomrule
        \end{widetable} 
        \begin{tablenotes}[flushleft]
  \setlength\labelsep{0pt}
    \footnotesize
    \item Notes: The table reports estimates from the DID regression of log fares on treated $\times$ post and market characteristics including the share of nonstop flights and the average distance.  We consider three mergers: Delta-Northwestern in Column (1), United-Continental in Column (2), and American-US in Column (3).
    \end{tablenotes}
    \end{threeparttable}
\end{table}

%% file: Tables/demand.tex
\begin{table}[H] \centering 
  \caption{Demand Estimates} 
  \label{tab:demand}
  \begin{threeparttable}
    \begin{widetable}{.8\columnwidth}{lc}
   \toprule
                                    & $\log(s_{jt})$ - $\log(s_{0t})$ \\    
   \midrule 
   Average Fare                     & -0.0048$^{***}$\\   
                                    & (0.0004)\\   
   $\log(S_t)$                        & 0.8356$^{***}$\\   
                                    & (0.0133)\\   
   Share Nonstop                    & 0.4030$^{***}$\\   
                                    & (0.0282)\\   
   Average Distance (1,000's)       & -0.4881$^{***}$\\   
                                    & (0.0498)\\   
   Average Distance$^2$ (1,000's)    & 0.0485$^{***}$\\   
                                    & (0.0045)\\   
   log(1 + Num. Fringe)             & -0.2642$^{***}$\\   
                                    & (0.0057)   
    \\ \midrule
   R$^2$                            & 0.94238\\  
   Observations                     & 1,283,472\\  
   Own-price elasticity             & -5.1652 \\  
    \\
   Origin-destination fixed effects & $\checkmark$\\   
   \bottomrule
\end{widetable}
\begin{tablenotes}[flushleft]
  \setlength\labelsep{0pt}
    \footnotesize
    \item Notes: The table presents results from demand estimation. Prices and log-within-shares $\log(S_t)$ are instrumented with the average rival distance, the average number of markets a rival serves, and the number of rival carriers. We include origin-destination fixed effects and cluster standard errors at the origin-destination level. The median own-price elasticity is $-5.1652$ which is in line with the literature.
    \end{tablenotes}
    \end{threeparttable}
\end{table}

%% file: Tables/supply.tex
\begin{table}[H] \centering 
    \caption{Bertrand-Implied Marginal Cost Estimates} 
    \label{tab:supply} 
    \begin{threeparttable}
    \begin{widetable}{.75\columnwidth}{lc}
       \toprule
                                        & Marginal Cost\\  
       \midrule 
       Average Distance (1,000's)       & 63.17$^{***}$\\   
                                    & (0.9502)   
        \\ \midrule
       R$^2$                            & 0.42757\\  
       Observations                     & 1,283,472\\  
        \\
       Origin-destination fixed effects & $\checkmark$\\   
       \bottomrule
    \end{widetable}
    \begin{tablenotes}[flushleft]
  \setlength\labelsep{0pt}
    \footnotesize
    \item Notes: The table presents results from supply estimation. We include origin-destination fixed effects and cluster standard errors at the origin-destination level.
    \end{tablenotes}
    \end{threeparttable}
\end{table}